\documentclass{vldb}
\usepackage{graphicx}
\usepackage{balance}  
\usepackage{url}
\usepackage{times}
\usepackage{color}
\usepackage{algorithm}
\usepackage{algorithmic}
\usepackage{amsmath}
\usepackage{flushend}
\usepackage{enumitem}
\usepackage{textcomp}
\newtheorem{theorem}{Theorem}[section]
\newtheorem{corollary}{Corollary}[section]

\setitemize{noitemsep,topsep=3pt,parsep=3pt,partopsep=0pt,leftmargin=12pt}
\setenumerate{noitemsep,topsep=3pt,parsep=3pt,partopsep=0pt,leftmargin=12pt}
\begin{document}

\title{Aggregate Estimation Over Dynamic Hidden Web Databases}
\numberofauthors{1}
\author{Weimo Liu$^\dag$, Saravanan Thirumuruganathan$^\ddag$, Nan Zhang$^\dag$, Gautam Das$^\ddag$\\\affaddr{The George Washington University$^\dag$; University of Texas at Arlington$^\ddag$}}

\maketitle

\begin{abstract}

Many databases on the web are ``hidden'' behind (i.e., accessible only through) their restrictive, form-like, search interfaces. Recent studies have shown that it is possible to estimate aggregate query answers over such hidden web databases by issuing a small number of carefully designed search queries through the restrictive web interface. A problem with these existing work, however, is that they all assume the underlying database to be static, while most real-world web databases (e.g., Amazon, eBay) are frequently updated. In this paper, we study the novel problem of estimating/tracking aggregates over dynamic hidden web databases while adhering to the stringent query-cost limitation they enforce (e.g., at most 1,000 search queries per day). Theoretical analysis and extensive real-world experiments demonstrate the effectiveness of our proposed algorithms and their superiority over baseline solutions (e.g., the repeated execution of algorithms designed for static web databases).
\end{abstract}

\section{Introduction} \label{sec:intro}

In this paper, we develop novel techniques for estimating and tracking various types of aggregate queries, e.g., COUNT and SUM, over {\em dynamic} web databases that are hidden behind proprietary search interfaces and frequently changed.

\noindent{\bf Hidden Web Databases:} Many web databases are ``hidden'' behind restrictive search interfaces that allow a user to specify the desired values for one or a few attributes (i.e., form a conjunctive {\em search query}), and return to the user a small number (bounded by a constant $k$ which can be 50 or 100) of tuples that match the user-specified query, selected and ranked according to a proprietary scoring function. Examples of such databases include Yahoo!~Autos, Amazon.com, eBay.com, CareerBuilder.com, etc.

\noindent{\bf Problem Motivations:} The problem we consider in this paper is how a third party can use the restrictive web interface to estimate and track {\em aggregate query answers} over a dynamic web database.
Aggregate queries are the most common type of queries in decision support systems as they enable effective analysis to glean insights from the data.

\begin{itemize}
\item Tracking the number of tuples in a web database is by itself an important problem. For example, the number of active job postings at Monster.com or listings at realestate.com can provide an economist with real-time indicators of US economy. Similarly, tracking the number of apps in Apple's App Store and Google Play provides us a continuous understanding of the growth of the two platforms. Note that while some web databases (e.g., App Store) periodically publish their sizes for advertisement purposes, such published size is not easily verifiable, and sometimes doubtful because of the clear incentive for database owners to exaggerate the number.
\item More generally, there is significant value in monitoring a wide variety of aggregates. For example, a sudden drop on the COUNT of used Ford F-150s in a used car database may indicate a subsequent increase of prices. Similarly, a rapid increase of AVG salary offered on job postings which require a certain skill (e.g., Java) may indicate an expansion of the corresponding market.
\item A number of economically important aggregates such as changes in employment month-over-month or the stock prices of specific sectors of companies (e.g., semiconductor stocks) fluctuate wildly. Tracking these aggregates effectively is of paramount interest to policy makers. In general, for any database, aggregates that are not overly broad (to the degree of COUNT(*) or AVG over the entire database) tend to change rapidly, with higher change frequencies for narrower aggregates. 
\end{itemize}

\noindent{\bf Challenges:} Since {\em aggregate queries} are not directly supported by the web interface, one can only estimate them by combining the results of multiple {\em search queries} that are supported. Prior work has shown ways to translate aggregate queries over a {\em static} hidden database to a small number of search queries, and then generate unbiased estimations of SUM and COUNT aggregates from the search query answers. But no existing technique supports dynamic databases that change over time. A seemingly simple approach to tackle the dynamic case is to repeatedly execute (at certain time interval) the existing ``static'' algorithms \cite{DJJ+10}. Unfortunately, this approach has two critical problems:
\begin{itemize}
\item Many real-world web databases limits the number of search queries one can issue through
per-IP (for web interface queries) or per-developer key (for API based queries) limits. In many cases, this daily limit is too low to sustain a complete execution of the static algorithm (to reach a reasonable accuracy level). For example, eBay limits API calls to 5,000 per day, making it extremely difficult for a third party user to track changes that occur to the database\footnote{\small{Note that while some web databases do display a list of recently added items (e.g., ``New and Noteworthy'' for App Store), such a list is rarely comprehensive. In addition, there is no direct way to identify tuples that were recently deleted (e.g., cancelled job listings at monster.com) or updated - which are important for understanding how the database changes over time.}}.
The existing (static) techniques handle this by stretching its execution across multiple days, and assume that the database does not change within the execution period - an assumption that often does not hold for frequently updated web databases such as eBay.
\item Even when the daily limit is high enough, repeated executions actually wastes a lot of search queries. To understand why, consider an extreme-case scenario where the underlying database remains unchanged. With repeated execution, the estimation error remains the same as the first-day estimate, even after numerous queries have been issued in later days. On the other hand, it is easy to see that if one somehow detects the fact that the database changes little over time, then all queries issued afterwards can be used to {\em improve} the estimation accuracy and reaching a significantly lower error than the simple ``repeated execution'' strategy.
\end{itemize}

Another straightforward approach is to track all changes that occur to the underlying database - i.e., to determine which tuples got inserted/deleted - and then use these changes to update the previous aggregate estimations. This approach, however, likely requires an extremely large number of queries because, as shown in previous studies of web database crawling \cite{SZTJ12}, the crawling of changed tuples through the web interface requires a prohibitively high query cost for most real-world settings.

\noindent{\bf Problem of Dynamic Aggregate Estimation:} In this paper, we initiate a study of estimating and tracking aggregates over a dynamic hidden web database. Note that the definition of ``aggregate'' here is more subtle than the static case. To understand why, note that any aggregate estimation process may require multiple (search) queries to be issued through the web interface. The issuance of these queries take time - during which the database and the aggregate might have already changed, making the estimation of it an ill-defined problem.

As such, for the ease of formally studying the problem of aggregate estimation over dynamic hidden databases, we introduce a theoretic concept ``{\em round}'' (i.e., a fixed time interval). Specifically, we assume that the database only changes at the starting instant of each round. With this definition, our objective is then (well) defined to be estimating an aggregate over the database state at the current round. Note that this ``round'' concept is introduced for theoretical purposes only, because many real-world databases are updated at arbitrary time/intervals - nevertheless, as we shall show in \S\ref{sec:sim}, results in the paper can be easily extended to address these cases.

\noindent{\bf Outline of Technical Results: } Aiming to address the shortcomings of the above-described ``repeated execution'' strategy (a strategy which we refer to as RESTART-ESTIMATOR), our first result is REISSUE-ESTIMATOR: Instead of restarting the process all over again at every round, we try to {\em infer} whether and how search query answers received in the last round change in this round (and revise the aggregate estimations accordingly).

It might seem as if one has to reissue every query issued in the last round to obtain its updated answer. But a key rationale for REISSUE-ESTIMATOR is that this is indeed not the case. Specifically, we find that the query cost required for detecting/inferring changes is far lower because of the interdependencies between different query answers. For example, often times confirming that one query answer does not change is enough for us to infer that multiple search queries must have stayed unchanged in the last round. This leads to a significant saving of query cost.

Nonetheless, it is important to note that, besides the reduction of query cost, we have another important objective: improving the accuracy of aggregate estimations. Here it might appear that REISSUE-ESTIMATOR has a major problem - specifically, unlike RESTART-ESTIMATOR which accesses / accumulates new ``sample points'' round after round and thereby expands its reach to a ever-increasing portion of the database, the reissuing idea can be stuck with the same, sometimes-unlucky, picks generated in the first few rounds because of its obligation on updating historic query answers. Intuitively, this could make it very difficult for REISSUE-ESTIMATOR to correct initial estimation errors (i.e., ``unlucky picks'') and converge to an accurate estimation.

Somewhat surprisingly, we found through theoretical analysis (and verified through experiments over real-world datasets) that, in practice, the exact opposite is usually true. That is, unless the database is almost regenerated after each round - e.g., with all old tuples removed and entirely new tuples inserted, in which case restarting is {\em sometimes} better - reissuing almost always offers a better tradeoff between estimation accuracy and query cost. We derive rigid conditions for the comparison between reissuing and restarting when different types of aggregates are being estimated, and verify their correctness through experiments.

While the idea of query reissuing significantly improves the performance of aggregate estimation over dynamic hidden databases, we find that REISSUE-ESTIMATOR still wastes queries in certain cases - specifically, when the database undergoes little change in a round. For many real-world databases that do not see frequent updates, this could lead to a significant waste of queries (and lackluster accuracy) after a period of time.

We develop RS-ESTIMATOR to address this problem. Our idea is close in spirit to reservoir sampling, which automatically maintains a sample of a database according to how the database changes. Specifically, the fewer changes happen to the database, the few changes the sample will see. We show through theoretical analysis and experimental studies that RS-ESTIMATOR further outperforms REISSUE-ESTIMATOR, especially in cases where changes to the database are small and/or infrequent.

In summary, the main contributions of this paper are as follows.
\begin{itemize}
\item We initiate the study of aggregate estimation over a {\em dynamic} web database through its restrictive search interface, so as to enable a wide range of applications such as tracking the number of active job postings at monster.com.
\item We propose a query-reissuing technique which leverages the historic query answers to produce up-to-date aggregate estimations without bias. We prove that, for many aggregates, our REISSUE-ESTIMATOR significantly outperforms the repeated execution of static algorithms.
\item We also propose a bootstrapping-based technique to auto-adjust the query plan according to how much the database has changed. The resulting algorithm, RS-ESTIMATOR, further improves upon REISSUE-ESTIMATOR especially when the database is updated at various intervals.
\item We provide a thorough theoretical analysis and experimental studies over synthetic and real-world datasets, including live experiments on real-world websites such as Amazon.com and eBay.com, that demonstrate the effectiveness of our proposed approaches.
\end{itemize}

\noindent {\bf Paper Organization:} The rest of this paper is organized as follows. In \S2 we introduce preliminaries and  problem definition. \S3 and \S4 are devoted to the development of REISSUE- and RS-ESTIMATOR, respectively. In \S5 we discuss a number of system design issues - e.g., how to deal with real-world web databases with various query limitations. \S6 contains a detailed experimental evaluation of our proposed approaches. \S7 discusses related work, followed by the conclusion in \S8.

\section{Motivation}

In this section, we start with introducing the model of hidden databases, their web interfaces, and how they change over time. Then, we define the objectives of aggregate estimation/tracking.

\subsection{Model of Dynamic Hidden Web Databases}
\noindent {\bf Hidden Web Database and Query Interface:} Consider a hidden database $D$ with $m$ attributes $A_1$, $\ldots, A_m$. Let $U_i$ be the domain for attribute $A_i$. For a tuple $t \in D$, we use $t[A_i] \in U_i$ to denote the value of $A_i$ for $t$. In this paper, we focus on categorical attributes, as numerical attributes can be discretized accordingly. We also assume all tuples to be distinct with no NULL values.

To query the hidden database, a user specifies the desired values for a subset of attributes through a conjunctive query of the form $q$: SELECT * FROM $D$ WHERE $A_{i_1}=u_{i_1}$ AND $\ldots$ AND $A_{i_s} = u_{i_s}$ where $i_1, \ldots, i_s \in [1,m]$  and $u_{i_j} \in U_{i_j}$. Let $Sel(q) \subseteq D$ be the tuples matching $q$. When there are more than $k$ (where $k$ is a constant fixed by the hidden database) matching tuples (i.e., $|Sel(q)| > k$), these tuples are ordered according to a proprietary {\em scoring function}, and then only the top-$k$ tuples are returned to the user.  We call such a query as an {\em overflowing} query. We say that a query is {\em valid} if it returns between 1 and $k$ tuples. A query {\em underflows} if it is too restrictive and returns empty.

\noindent {\bf Dynamic Hidden Databases:} Real-world hidden web databases change over time.
In most part of the paper, we consider a {\em round-update} model where modifications occur at the beginning instant of each round. This model is introduced purely for ease of theoretical analysis. Then, in \S\ref{sec:sim}, we extend the notion to arbitrary updates where we make {\em no assumption} about what frequency, or at what time, a hidden database is updated. Following the round-update model, we set round to be a short interval (e.g., one day) during which the database is considered fairly static.  We denote the $i$-th round as $\mathcal{R}_i$, and represent the database state during Round $\mathcal{R}_i$ as $D_i$.

Most web databases impose a per user/per IP limit on the number of queries that can issued over a time frame. Let $G$ be the number of queries one can issue to the database per round. Note that if the query limit is enforced at a different time interval (e.g., per second), we can always extrapolate the limit according to the length of a round and set $G$ accordingly.

It is important to note that we assume no knowledge of which tuples have been inserted to and/or deleted from the hidden database in a given time period. That is, we consider the worst-case scenario where each tuple is {\em not} associated with a timestamp indicating the last time it has been inserted or changed. The key rationale behind making this worst-case assumption is the complexity of timestamps provided by real-world web databases. For example, while Amazon discloses the first date a product was added to Amazon.com, it does not disclose the last time an attribute (e.g., Price) was changed, or a variation (e.g., a new color) was added. iOS App Stores have a similar policy where software updates are timestamped but price and/or description changes are not. Because of these complex policies, for the purpose of this paper, we do not consider the usage of timestamps for analyzing hidden web databases.

\subsection{Objectives of Aggregate Estimation} \label{sec:oae}

In this paper, we consider two types of aggregate estimation tasks over a dynamic hidden database:
\begin{itemize}
\item Single-round aggregates: At each round $\mathcal{R}_i$, estimate an aggregate query of the form $Q(D_i)$: SELECT AGG($f(t)$) FROM $D_i$ WHERE {\em Selection Condition}, where AGG is the aggregate function (we consider SUM, COUNT, and AVG in this paper), $f(t)$ is any function over (any attribute or attributes of) a tuple $t$, and {\em Selection Condition} is any condition that can be independently evaluated over each tuple (i.e., there is a deterministic function $g(t)$ such that $g(t) = 1$ if tuple $t$ satisfies the selection condition and $0$ otherwise).

\item Trans-round aggregates: At each round $\mathcal{R}_i$, estimate an aggregate over data from both the current round and the previous ones of the form $Q$: SELECT AGG($f(t(D_1)$, $\ldots$, $t(D_i))$) FROM $D_1, \ldots, D_i$ WHERE {\em Selection Condition} - where $t(D_j)$ is the value of a tuple $t$ in Round $\mathcal{R}_j$ (if the tuple indeed exists in $D_j$). An example is the change of database size from $D_{i-1}$ to $D_i$ - in which case $f(t(D_{i-1}), t(D_i)) = f^\prime(t(D_i)) - f^\prime(t(D_{i-1}))$ where $f^\prime(t(D_j)) = 1$ if $t \in D_j$ and $0$ otherwise.
\end{itemize}

We would like to note that, while some trans-round aggregates (e.g., change of database size) can be ``rewritten'' as a combination of single-round aggregates, many cannot. For example, no combination of single-round aggregates can answer a trans-round aggregate ``the average price drop for all products with a lower price today than yesterday''. As we shall show in \S\ref{sec:rie}, even for those trans-round aggregates that can be rewritten, directly estimating a trans-round aggregate may yield significantly more accurate results than post-processing single-round estimates.

For both categories, we aim to estimate the answer to an {\em aggregate query} by issuing a small number of {\em search queries} through the restrictive web interface. For most part of the paper, we focus on tracking a fixed aggregate query over multiple rounds, and extend the results in \S\ref{sec:sim} to handling ad-hoc aggregate queries (e.g., how to estimate the change of database size from $\mathcal{R}_1$ to $\mathcal{R}_2$ if we receive this query after $\mathcal{R}_1$ is already passed).

For a given aggregate, the goal of aggregate estimation is two-fold. One is to maintain query cost per round below the database-imposed limit $G$, while the other is to minimize the estimation error. For measuring such estimation error, note the two components composing the error of estimation $\tilde{\theta}$ for an aggregate $\theta$: {\em Bias} $E(\tilde{\theta} - \theta)$ where $E(\cdot)$ is expected value taken over randomness of $\tilde{\theta}$, and {\em Variance} of $\tilde{\theta}$. Specifically, $\tilde{\theta}$'s mean squared error (MSE) is
\begin{align}
\text{MSE}(\widetilde{\theta}) = \text{Bias}^2(\widetilde{\theta}) + \text{Variance}(\widetilde{\theta}).
\end{align}
Thus, in terms of minimizing estimation error, our goal is to reduce both estimation bias and variance.

\section{REISSUE-ESTIMATOR} \label{sec:rie}

Recall from \S1 a ``restarting'' baseline for aggregate estimation over dynamic hidden database which treats every round as a separate (static) hidden database and reruns a static aggregate estimation algorithm \cite{DJJ+10} at each round to produce an independent estimation. While simple, this restarting baseline wastes numerous queries because no information is retained/reused from round to round. Our goal in this section is to develop REISSUE-ESTIMATOR, an algorithm that significantly outperforms the baseline by leveraging historic query answers. Specifically, we first introduce the main idea of query reissuing, and then present theoretical analysis for a key question: is reissuing or restarting better for tracking aggregates when the database has varying degrees of change?

\noindent{\em Running Example:} We use Figure~\ref{fig:qtr} (explained below) as a running example throughout this section.

\subsection{Query Reissuing for Multiple Rounds} \label{sec:ffr}
For the ease of developing the idea of query reissuing and understanding its difference with the restarting baseline, we start this subsection by presenting a simple interpretation of the static aggregate-estimation techniques developed in the literature \cite{DJJ+10}. Then, we describe the idea of query reissuing and explain why it may lead to a significant saving of query cost.

\noindent{\bf Simple Model of Static Aggregate Estimation:} Existing algorithms for single-round aggregate estimation are centered on an idea of {\em drill downs} over a query tree depicted in Figure~\ref{fig:qtr}. The query tree organizes queries from broad to specific on top to bottom. Specifically, the root level is SELECT * FROM $D$. Going each level deeper, the query is appended by a conjunctive predicate formed by (1) the attribute corresponding to the level (Level $i$ is corresponding to Attribute $A_i$), and (2) the attribute domain value corresponding to the branch - e.g., in Figure~\ref{fig:qtr}, the three second-level nodes, from left to right, are corresponding to queries SELECT * FROM D WHERE $A_1 = u_1$, WHERE $A_1 = u_2$, and WHERE $A_1 = u_3$, respectively. One can see that each leaf node of the tree is corresponding to a fully-specified, $m$-predicate, conjunctive query. There are $\prod^m_{i=1} |U_i|$ such nodes. As such, we number the leaf nodes from $1$ to $\prod^m_{i=1} |U_i|$, respectively.

\begin{figure} [ht]
\centering
\includegraphics[width = 2.8in]{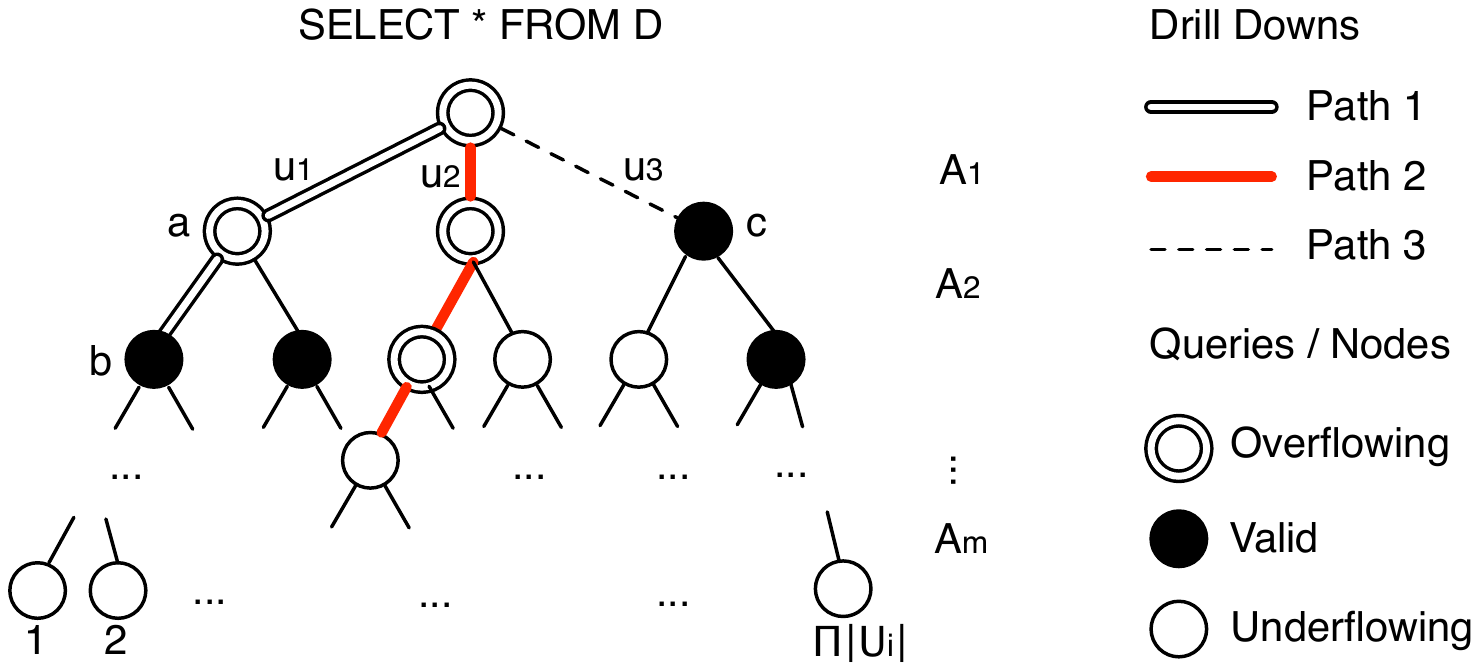}
\vspace{-3mm}
\caption{Query Tree}
\label{fig:qtr}
\vspace{-3mm}
\end{figure}

Given the query tree structure, existing single-round aggregate estimation techniques perform a simple {\em drill down} process that can be understood as follows: First, choose a leaf node No.~$j$ uniformly at random from $[1, \prod^m_{i=1} |U_i|]$. Then, find the path from root to Leaf No.~$j$ and issue each node (i.e., query) on the path one by one from top to bottom (i.e., {\em drill down}) until finding a node that does {\em not} overflow. For example, if we happen to choose Leaf No.~1 (i.e., Path 1) in the running example, we would issue three queries, the root, query a, and query b, and stop because query b is valid. As another example, if Leaf No.~$\prod^m_{i=1} |U_i|$ is chosen, we would find Path 3 and issue two queries, the root and query c.

This {\em top non-overflowing} query $q$ is then used to produce an unbiased estimation for a SUM or COUNT aggregate $Q$: SELECT AGG($f(t)$) FROM D WHERE {\em Selection Condition}. Specifically, the estimation is $\tilde{Q}$ = $Q(q)/p(q)$, where the numerator $Q(q)$ is the result of applying AGG over tuples returned by $q$ that satisfy the selection condition, and $p(q)$ is the ratio of leaf nodes that have $q$ on its path - i.e., the probability for $q$ to be issued during a drill down instance. For example, each of the three Level-2 nodes in the running example has $p(q) = 1/3$. Thus, the estimation of COUNT(*) produced by Path 3 is $3c$, where $c$ is the number of tuples returned by query c. Similarly, if Path 2 is taken, the aggregate estimation generated for any SUM or COUNT query will be 0 because the path terminates at an underflowing (i.e., empty) node.

It was proved in \cite{DJJ+10} that such an estimation is {\em unbiased}. The reason is easy to understand: Each tuple is returned by one and only one top non-overflowing query in the tree. In other words, for any SUM or COUNT aggregate $Q$ (with or without selection conditions), the sum of $Q(q)$ for all top non-overflowing queries is always $Q$. As such, the expected value of $Q(q)/p(q)$ is $\sum_q p(q) \cdot Q(q)/p(q) = Q$ (i.e., it is an unbiased estimator). To reduce estimation variance, one can repeatedly perform multiple drill downs (until exhausting the database-imposed query limit $G$) and use the average output as the final estimation.

\noindent{\bf Query Reissuing: } One can see from the simple model that the randomness of each drill-down can be captured by a single random number $r$ - i.e., sequence number of the leaf-level node corresponding to the drill-down. As such, each execution of the single-round algorithm can be uniquely identified by a ``{\em signature}'' set $\mathcal{S} = \{r_1, \ldots, r_h\}$ where each $r_i$ defines one drill-down performed.

With this notation, a key difference between restarting and reissuing can be stated as follows. With the restarting baseline, at each round, say round $t$, we randomly generate a signature set $\{r^t_1, \ldots, r^t_s\}$ (with a sufficiently large $s$), and perform drill-downs $r^t_i$ in an increasing order of $i$ until exhausting the query limit. With reissuing, on the other hand, we generate $\{r^t_1, \ldots, r^t_s\}$ only once when $t = 1$, and then reuse the same set for every subsequent round (i.e., with query reissuing, $\forall t, t^\prime, i$, there is $r^t_i = r^{t^\prime}_i$).

\noindent{\bf Saving of Query Cost: } To understand why reusing the same signature set can lead to a significant saving of query cost, note that for each drill-down $r^1_i$, before performing it at Round $\mathcal{R}_j$, we already know substantial information about it from the previous rounds. Specifically, we know that at Round $\mathcal{R}_{j-1}$, which query $q$ is the top non-overflowing node for $r^1_i$. Thus, for Round $\mathcal{R}_j$, we can start our drill-down from $q$ instead of the root node, knowing that all queries above $q$ are likely to overflow. In the running example, this means to start from query b instead of the root for ``updating'' Path 1 in a new round. After issuing $q$, we can drill down further if $q$ overflows, ``roll up'' to the first non-underflowing node if $q$ underflows, and, if $q$ is valid, directly generate an estimation. One can see that this may lead to a significant saving of query cost, especially when the database undergoes little change. For example, in the running example, if no change happens to the database, we can save 2 out of 3 queries (i.e., root and query a) for updating Path 1.

It is important to note that the saving of query cost can be directly translated to more accurate aggregate estimations. In particular, since updating a drill-down may consume fewer queries, the remaining query budget (after updating all previous drill downs) can be used to initiate new drill downs, increasing the number of drill downs and thereby reducing the estimation variance and error. Algorithm~\ref{alg:reissueEst} (at the end of this section) depicts the pseudocode for REISSUE-ESTIMATOR.

\noindent{\bf Unbiasedness of Aggregate Estimations:} An important property of REISSUE-ESTIMATOR is that it produces unbiased estimations on all COUNT and SUM aggregates with or without selection conditions (as defined in \S\ref{sec:oae}). This is shown in the following theorem.

\begin{theorem} \label{thm:ubr}
At any round, for a given SUM or COUNT aggregate $Q$, we have $E(\pi(r_i)) = Q$, where $\pi(r_i)$ is the estimation produced by REISSUE-ESTIMATOR from drill down $r_i$, and $E(\cdot)$ is the expected value taken over the randomness of $r_i$.
\end{theorem}

\begin{proof}
Note from the design of REISSUE-ESTIMATOR that the COUNT(*) estimation produced by $r_i$ is
\begin{align}
\pi(r_i) = \frac{|q(r_i)|}{p(q(r_i))}
\end{align}
where $q(r_i)$ is the first non-overflowing one on the path from root to leaf No.~$r_i$, $|\cdot|$ is the number of tuples it returns, and $p(q(r_i))$ is the ratio (within all possible signatures) of signatures for which $q(r_i)$ appears on the path from root to leaf. Let $S$ be the set of all possible signatures. We have
\begin{align}
E(\pi(r_i)) &= \sum_{r \in S} \frac{1}{|S|} \cdot \frac{|q(r)|}{p(q(r))}\\
&= \sum_{q | \exists r \in S, q(r) = q} \frac{|\{r : r \in S, q(r) = q\}|}{|S|} \cdot \frac{|q|}{p(q)}\\
&= \sum_{q | \exists r \in S, q(r) = q} \frac{|\{r : r \in S, q(r) = q\}|}{|S|} \cdot \nonumber \\
& \hspace{10mm} \frac{|q| \cdot |S|}{|\{r : r \in S, q(r) = q\}|}\\
&= \sum_{q | \exists r \in S, q(r) = q} |q| = |D_j|. \label{equ:crs}
\end{align}
The critical step (\ref{equ:crs}) holds because each tuple belongs to one and only one $q$. To understand why, note that no two $q$ can exist on the same path from root to leaf (otherwise at least one is not the first non-overflowing node on the path). The tree definition ensures that any such two $q$ must have mutually exclusive (selection conditions on) values. Thus, each tuple belongs to one and only one $q$.
\end{proof}

Note that similar to the static case, REISSUE-ESTIMATOR produces (slightly) biased estimations for AVG queries, because the ratio between unbiased SUM and COUNT estimators may not be unbiased.

\subsection{Key Question: Reissue or Restart?}

While the idea of query reissuing likely saves queries, it is important to understand another implication of reissuing - the characteristics of aggregate estimations also changes. Specifically, one can see that, with the reissuing method, estimations produced in different rounds are {\em no longer independent} of each other due to the reuse of the same signature set (of drill-downs). We now discuss whether such change of characteristics can lead to more (or less) accurate estimations for various types of aggregates when the database undergoes different degrees of change.

\subsubsection{Motivating Examples} \label{sec:moe}

To illustrate the fact that reissue and restart may each perform better in different cases, we start by discussing two extreme-case scenarios: (1) {\em no change} - the database never changes over time, and (2) {\em total change} - the database is randomly (re-)generated from a given distribution at every round. In both cases, we consider the estimation of a single-round aggregate SELECT COUNT(*) FROM $D_i$, which we denote as $|D_i|$, and a trans-round aggregate (SELECT COUNT(*) FROM $D_i$) $-$ (SELECT COUNT(*) FROM $D_{i-1}$), which we denote as $|D_i| - |D_{i-1}|$.

\noindent{\bf Example 1 (No change):} First, consider the case when the database never changes round after round. One can see that, for updating each drill down, the queries issued by REISSUE-ESTIMATOR are always a subset of those issued by RESTART-ESTIMATOR. Since such a saving of query cost allows REISSUE-ESTIMATOR to initiate new drill downs, the expected number of drill downs it performs (i.e., updates or initiates) in Round $\mathcal{R}_i$ is higher than that of RESTART-ESTIMATOR. This leads to a smaller variance and therefore a smaller error for estimating $|D_i|$.

For $|D_i| - |D_{i-1}|$, the error of REISSUE-ESTIMATOR is even smaller because every estimation generated from updating a previous drill down is exactly 0 and therefore has no error at all. If REISSUE-ESTIMATOR updates all $h_1$ previous drill downs and initiates $h_2$ new ones, then the estimation variance is $\sigma^2 \cdot h_2 / (h_1 \cdot (h_1 + h_2))$, where $\sigma^2$ is the variance of estimation produced by a single drill down.  Note that this is significantly smaller than $\sigma^2 / h + \sigma^2 / h^\prime$, the variance produced by RESTART-ESTIMATOR after performing $h$ drill downs in Round $\mathcal{R}_{i-1}$ and $h^\prime$ ones in Round $\mathcal{R}_i$. For example, even in a conservative scenario when $h_1 = h = h^\prime$ (note that there is often $h_1 \gg \max(h, h^\prime)$ because of the queries saved by reissuing), the variance produced by REISSUE-ESTIMATOR is at most 50\% of RESTART-ESTIMATOR, regardless of the value of $h_2$.


\noindent{\bf Example 2 (Total change):} When the database is completely regenerated at each round,  accuracy analysis is straightforward because estimations produced by reissuing and restarting a drill down are statistically indistinguishable. Thus, the only factor to consider is the query cost per drill down - the fewer an algorithm consumes the better, regardless of which aggregate is to be estimated.

For query cost analysis, we first note the existence of data distributions (according to which the database is regenerated at each round) with which REISSUE ends up consuming more queries than RESTART. For example, consider a distribution that guarantees if a tuple $t$ exists in the database, there is always $t^\prime$ which shares the same value with $t$ on attributes $A_1, \ldots, A_{m-1}$ but different ones on $A_m$. With this distribution, when $k = 1$, a drill down can terminate at Level-$(m - 1)$ of the tree. If we try to update this drill down, however, we start at Level-$(m-1)$ but might have to travel all the way up to the root to find a non-empty query in the new round because the database has been completely regenerated.  This leads to a significantly higher query consumption than a brand new drill down starting from the root. In other words, REISSUE-ESTIMATOR might end up performing worse than RESTART-ESTIMATOR.

It is important to note that there are also cases where REISSUE-ESTIMATOR is better. For example, consider a Boolean database with $n = 2^{m/2}$ tuples, each attribute of which is generated i.i.d.~with uniform distribution. With REISSUE-ESTIMATOR, the expected level from which a (reissued) drill down starts is Level-$m/2$, while RESTART-ESTIMATOR has to start with the root and consumes more queries (in an expected sense, when $m \gg 1$) per drill down.

\subsubsection{Theoretical Analysis} \label{sec:ta}

The above examples illustrate that the performance comparison between REISSUE- and RESTART-ESTIMATOR does not have a deterministic answer. Rather, it depends on the data distribution, how it changes over time, and also the aggregate to be estimated. We now develop theoretical analysis that indicates under which conditions will reissue or restart have a better performance.

Note that, given a previous drill down which terminates at query $q$, there are three types of updates REISSUE-ESTIMATOR might encounter in the new round:

\begin{enumerate}
\item If $q$ neither overflows nor underflows in a new round, then it is the only query REISSUE-ESTIMATOR issues to update the drill down - i.e., all queries on the path from the root to $q$ will be saved (compared with RESTART-ESTIMATOR).

\item If $q$ overflows in the new round, then REISSUE-ESTIMATOR starts from $q$ and drills down until reaching a valid or underflowing query. It still saves all queries on the path from the root to $q$ as compared with RESTART-ESTIMATOR.

\item Finally, if $q$ now underflows, REISSUE-ESTIMATOR has to trace upwards, along the path from $q$ to the root node, until finding a node that either overflows or is valid, say $q_\mathrm{t}$. In this case, the queries REISSUE-ESTIMATOR issues (and wastes) are on the path between $q$ and $q_\mathrm{t}$.
\end{enumerate}

One can see that the only case where REISSUE-ESTIMATOR {\em might} require more queries than RESTART-ESTIMATOR is Case (3). Thus, the worst-case scenario for REISSUE-ESTIMATOR is when $n_\mathrm{i} = 0$, i.e., when no tuple is inserted into the database. This worst-case scenario results in the following theorem.

\begin{theorem} \label{thm:nth}
After removing $n_\mathrm{d}$ randomly chosen tuples from an $n$-tuple database, for any SUM or COUNT query, the standard error of estimation produced by REISSUE-ESTIMATOR in the new database, denoted by $s_\mathrm{I}$, satisfies
\begin{align}
s_\mathrm{I} \leq \left(1 - \frac{n_\mathrm{d}}{n}\right) \cdot \sqrt{\frac{2 \max_{i \in [1, m]} (\log |U_i|) }{\log n - \log k} + \left(\frac{n_\mathrm{d}}{n}\right)^{k+1}} \cdot s_\mathrm{S}, \label{equ:t1m}
\end{align}
where $s_\mathrm{S}$ is the standard error of estimation produced by RESTART-ESTIMATOR in the old database, $|U_i|$ is the domain size of attribute $A_i$ ($i \in [1, m]$), and $k$ is as in top-$k$ interface. Specifically, when $n$ is sufficiently large to ensure an expected drill-down depth of at least 2 for RESTART-ESTIMATOR over the old database, we have $s_\mathrm{I} < s_\mathrm{S}$.
\end{theorem}

\noindent {\bf Proof.} 

\begin{proof}

We start by defining a set of {\em top-non-overflowing nodes} $\Lambda$.  Specifically, for each possible drill down path $r$, we find the highest-level (i.e., closest to root) node on $r$, say $q_r$, that does not overflow in the old round.  We then add $q_r$ to $\Lambda$.  One can see that, for every possible drill down path $r$, $\Lambda$ contains exactly one node on $r$. Thus, the expected number of queries issued by RESTART-ESTIMATOR in the old round is
\begin{align}
E[c_\mathrm{S}] = \sum_{q \in \Lambda} p(q) \cdot \ell(q). \label{equ:csv}
\end{align}
where $E[\cdot]$ stands for the expected value, $p(q)$ is the sum of selection probability for all drill down paths that pass through $q$ (note that such probability does not change for each round), and $\ell(q)$ is the expected number of queries issued by (i.e., expected depth of) all drill downs that pass through $q$ in the old round.

We now consider the query cost for REISSUE-ESTIMATOR in the new round. Specifically, we considers a subset of nodes in $\Lambda$, say $\Lambda^\prime$, which satisfy the following property: for each $q \in \Lambda^\prime$, the parent of $q$, say $\pi(q)$, changes from overflowing to underflowing in the transition.  The expected number of queries REISSUE-ESTIMATOR needs to update a drill down in the new round satisfies
\begin{align}
E[c_\mathrm{I}] \leq 2 + \sum_{q \in \Lambda^\prime} (p(q) \cdot \ell(q)). \label{equ:civ}
\end{align}

The reason can be stated as follows. First, observe from the definition of $\Lambda^\prime$ that, if Case (3) happens and REISSUE-ESTIMATOR issues {\em more than} 2 queries, the parent of $q$ must underflow because otherwise we would have stopped after issuing $q$ and the parent of $q$. This means that $q \in \Lambda^\prime$. Also note that the number of queries issued by REISSUE-ESTIMATOR in this case is at most $\ell(q)$, while in all other cases the query cost for updating a drill down is either 1 (Cases 1 and 2) or 2 (Case 3, the parent of $q$ does not underflow). Thus, we have Inequality (\ref{equ:civ}).

Let $\delta(q)$ be the probability for a top-non-overflowing query $q \in \Lambda$ to fall into $\Lambda^\prime$. According to (\ref{equ:civ}) and (\ref{equ:csv}) we have
\begin{align}
E[c_\mathrm{I}] \leq 2 + E[c_\mathrm{S}] \cdot \max_{q \in \Lambda} \delta(q).
\end{align}
Note that, in order for an overflowing query $q \in \Lambda$ to underflow in the new round, all tuples (at least $k + 1$ of them) returned by $q$ must be removed during the transition. Thus, $\forall q \in \Lambda$,
\begin{align}
\delta(q) \leq \left(\frac{n_\mathrm{d}}{n}\right)^{k+1}
\end{align}
In other words,
\begin{align}
E[c_\mathrm{I}] \leq 2 + E[c_\mathrm{S}] \cdot \left(\frac{n_\mathrm{d}}{n}\right)^{k+1}.
\end{align}

Note that for any SUM or COUNT query, $\sigma^2_\mathrm{I}$, the variance of estimation generated by a drill down updated by REISSUE-ESTIMATOR in the new round, and $\sigma^2_\mathrm{S}$, the variance of estimation generated by a drill down in the old round, satisfy $\sigma^2_\mathrm{I} \leq (1 - n_\mathrm{d}/n)^2 \cdot \sigma^2_\mathrm{S}$, for the simple reason that the expected query answer over each node (i.e., why applying the aggregate over all tuples returned by the node) in the new round is $1 - n_\mathrm{d}/n$ of its old value. As such, for a given query budget of $b$, the variance of final estimation produced by REISSUE-ESTIMATOR in Round 2 satisfies
\begin{align}
s^2_\mathrm{I} &\leq \frac{(1 - n_\mathrm{d}/n)^2 \cdot \sigma^2_\mathrm{S} \cdot (2 + E[c_\mathrm{S}] \cdot (n_\mathrm{d}/n)^{k+1})}{b}\\
&= \frac{2 \cdot (1 - n_\mathrm{d}/n)^2 \cdot s^2_\mathrm{S}}{E[c_\mathrm{S}]} + \left(1 - \frac{n_\mathrm{d}}{n}\right)^2 \cdot \left(\frac{n_\mathrm{d}}{n}\right)^{k+1} \cdot s^2_\mathrm{S} \label{equ:si1}\\
&\leq \left(1 - \frac{n_\mathrm{d}}{n}\right)^2 \cdot \left(\frac{2 \max_{i \in [1, m]} (\log |U_i|) }{\log n - \log k} + \left(\frac{n_\mathrm{d}}{n}\right)^{k+1}\right) \cdot s^2_\mathrm{S} \label{equ:si2}
\end{align}
Here the derivation from (\ref{equ:si1}) to (\ref{equ:si2}) is due to the following inequality 
\begin{align}
E[c_\mathrm{S}] \geq \log (n/k) / \log (\max_{i \in [1, m]} |U_i|),
\end{align}
which holds because for an $n$-tuple database with a top-$k$ interface and a maximum attribute
domain size of $\max|U_i|$, the average length of a drill down is at least $\log(n/k)/\log(\max|U_i|)$.
 As such, we have (\ref{equ:t1m}).

Note that when $n$ is sufficiently large to ensure $E[c_\mathrm{S}] \geq 2$, we have
\begin{align}
s^2_\mathrm{I} &\leq \left(1 - \frac{n_\mathrm{d}}{n}\right)^2 \cdot \left(1 + \left(\frac{n_\mathrm{d}}{n}\right)^{k+1}\right) \cdot s^2_\mathrm{S}\\
&< \left(1 - \frac{n_\mathrm{d}}{n}\right) \cdot \left(1 + \frac{n_\mathrm{d}}{n}\right) \cdot s^2_\mathrm{S}\\
&= \left(1 - \left(\frac{n_\mathrm{d}}{n}\right)^2\right) \cdot s^2_\mathrm{S} \leq s^2_\mathrm{S}.
\end{align}
Thus, $s_\mathrm{I}$ is always smaller than $s_\mathrm{S}$ in this case.

\end{proof}


\subsection{Algorithm REISSUE-ESTIMATOR} \label{sec:are}
In this subsection, we put all the previous discussions together to develop Algorithm REISSUE-ESTIMATOR. We then briefly describe how it can be extended to handle other aggregates and selection conditions. Algorithm~\ref{alg:reissueEst} depicts the pseudocode. We use the notation $q(r_i)$ to represent the first non overflowing query for drill down $r_i$ and let $\mathcal{Q} = \{q(r_1), \ldots, q(r_s)\}$. $parent(q(r_i))$ corresponds to parent of $q(r_i)$ in the query tree.

\begin{algorithm}[!htb]
\caption{{\bf REISSUE-ESTIMATOR}}
\begin{algorithmic}[1]
\label{alg:reissueEst}
\STATE Randomly generate signature set $\mathcal{S}=\{r_1, \ldots, r_s\}$
\STATE In $\mathcal{R}_1$, perform drill downs in $\mathcal{S}$ till exhausting query budget. Update $\mathcal{Q}$.
\FOR {each round $\mathcal{R}_i$}
\FOR {drill-down $r_j \in \mathcal{S}$}
\IF {$q(r_j)$ overflows}
\STATE Do drill-down from $q(r_j)$ till a non-overflowing node
\ELSIF{$q(r_j)$ underflows}
\STATE Do roll-up from $q(r_j)$ till a non-underflowing node or an underflowing node with an overflowing parent
\ENDIF
\ENDFOR
\STATE Issue new drill-downs from $\mathcal{S}$ for remaining query budget
\STATE Produce aggregate estimation according to $\mathcal{Q}$
\ENDFOR
\end{algorithmic}
\end{algorithm}

%
%
%
\vspace{1mm}
\noindent {\bf Aggregate Estimation with Selection Conditions:} So far, we have focused on aggregate queries that select all tuples in the database. To support aggregates with selection conditions - e.g., the number of science fiction books in Amazon - REISSUE-ESTIMATOR only needs a simple modification. Recall that a selection condition is specified via conjunctive constraints over a subset of attributes. Given the selection conditions $Q$, we simply alter our query tree to be the subtree corresponding to $Q$. In the Amazon example, we construct a query tree with all queries containing predicate Type = science-fiction. With the new query tree, drill-downs can be computed directly over it to a get an unbiased estimate.


\section{RS-ESTIMATOR} \label{sec:sam}

In this section, we start by discussing the problem of REISSUE-ESTIMATOR, especially its less-than-optimal performance when the database undergoes little change at a round. Then, we introduce our key ideas for addressing the problem, and present the detailed description of RS-ESTIMATOR.

\subsection{Problem of REISSUE-ESTIMATOR}

To understand the problem of REISSUE-ESTIMATOR, consider an extreme-case scenario where the database does not change round after round. One can see that, since updating each drill down requires exactly two queries (one query to verify if the result is same and one query over its parent to determine if it is still the top non-overflowing query), the number of drill downs that can be updated is at most $G/2$, where $G$ is the total query budget per round. This essentially places a lower bound on the estimation variance achieved by REISSUE-ESTIMATOR for many queries - e.g., when the aggregate query to be estimated is COUNT(*), the estimation variance has a lower bound of $2\sigma^2/G$, where $\sigma^2$ is the estimation variance produced by a single drill down (for the COUNT(*) query), no matter how many rounds have passed since the database last changed.

Note that this lower bound indicates a significant waste of queries by REISSUE-ESTIMATOR. In particular, it means that after a sufficient number of rounds, issuing more queries ($G$ per round for whatever number of rounds) will not further reduce the estimation error. Intuitively, it is easy to see that solutions exist to reduce this waste. Specifically, one does not need to issue many queries (i.e., update many drill downs) before realizing the database has changed little, and therefore reallocate the remaining query budget to initiate new drill downs (which {\em will} further reduce the estimation error).

We note that this intuition has indeed been used in statistics - e.g., in the design of reservoir sampling. With reservoir sampling \cite{vitter85}, how much change should happen to the sample being maintained depends on how much incoming data are inserted to the database. Similarly, we can determine the number of drill downs to be updated based on an understanding of how much change has occurred to the underlying database.

Unfortunately, reservoir sampling has two requirements which we do have the luxury to satisfy over a web database: (1) it requires the database to be insertion-only, while web databases often feature both insertions and deletions, and (2) reservoir sampling needs to know {\em which tuples} have been inserted. In our case, it is impossible to know what changes occurred to the database without issuing an extremely large number of queries to essentially ``crawl'' the database. Next, we shall describe our main ideas for solving these problems and enable a reservoir-like estimator.

\subsection{Key Ideas of RS-ESTIMATOR} \label{sec:kir}

\vspace{1mm}
\noindent {\bf Overview of RS-ESTIMATOR: } The key idea of RS-ESTIMATOR is to distribute the query budget available for each round into two parts: one for reissuing (i.e., updating) drill downs from previous rounds, and the other for initiating new drill downs. How the distribution should be done depends on how much changes have occurred - intuitively, the smaller the aggregate (to be estimated) changes, the fewer queries should RS-ESTIMATOR use for reissuing and the more should it use for new drill downs.

Thus, in order to determine a proper query distribution, RS-ESTIMATOR first uses a small number of {\em bootstrapping} queries to estimate the amount of change to the aggregate. To illustrate how this can be done, consider an example where the aggregate to be estimated is COUNT(*), and we are at Round $\mathcal{R}_2$ with $h$ drill downs performed in Round $\mathcal{R}_1$. In the following discussion, we first briefly describe how RS-ESTIMATOR processes the two types of drill downs to estimate $|D_2|$ once the distribution is given, and then present theoretical analysis for the optimal distribution.
Algorithm~\ref{alg:rsEst} provides the pseudocode for RS-ESTIMATOR.

\vspace{1mm}
\noindent{\bf Aggregate Estimation from Two Types:} Let $h_1$ be the number of drill downs (among the $h$ drill downs $r_1, \ldots, r_h$ performed in Round $\mathcal{R}_1$) we update, and $h_2$ be the number of new ones we perform. Without loss of generality, let the signatures for these $h_1$ and $h_2$ drill downs be $r_1, \ldots, r_{h_1}$ and $r^\prime_1, \ldots, r^\prime_{h_2}$, respectively. With RS-ESTIMATOR, we first use the $h_1$ updated drill downs to estimate $|D_2| - |D_1|$. Specifically,
\begin{align}
\tilde{v}_\mathrm{c} = \frac{1}{h_1} \cdot \sum_{i \in [1, h_1]} \frac{|q_2(r_i)|}{p(q_2(r_i))} - \frac{|q_1(r_i)|}{p(q_1(r_i))},
\end{align}
where $q_1(r_i)$ and $q_2(r_i)$ is the top non-overflowing query for path $r_i$ in Round $\mathcal{R}_1$ and $\mathcal{R}_2$, respectively, and $p(q)$ represents the ratio within all possible signatures for which $q$ is the top non-overflowing query. In addition, RS-ESTIMATOR also uses the $h_2$ new drill downs to produce an estimation of $|D_2|$, i.e.,
\begin{align}
\tilde{v}_\mathrm{d} = \frac{1}{h_2} \cdot \sum_{i \in [1, h_2]} \frac{|q_2(r^\prime_i)|}{p(q_2(r^\prime_i))}.
\end{align}
Note that RS-ESTIMATOR now has two independent estimations for $|D_2|$: (1) $\tilde{v}_1 + \tilde{v}_\mathrm{c}$, where $\tilde{v}_1$ is the first-round estimation of $|D_1|$, and (2) $\tilde{v}_\mathrm{d}$. A natural idea is to generate the final estimation as a weighted sum of the two - i.e., $\tilde{v}_2 = w_1 \cdot (\tilde{v}_1 + \tilde{v}_\mathrm{c}) + (1-w_1) \cdot \tilde{v}_\mathrm{d}$, where $w_1 \in [0, 1]$. The following theorem establishes the unbiasedness of the estimation no matter what $w_1$ is.

\begin{theorem} \label{thm:uep}
If $\tilde{v}_1$ is an unbiased estimation of $|D_1|$, then
\begin{align}
E(\tilde{v}_1 + \tilde{v}_\mathrm{c}) = E(\tilde{v}_\mathrm{d}) = |D_2|.
\end{align}
where $E(\cdot)$ denotes expected value taken over the randomness of $r_1, \ldots, r_{h_1}, r^\prime_1, \ldots, r^\prime_{h_2}$.
\end{theorem}

The proof follows in analogy to that of Theorem~\ref{thm:ubr}. Since $w_1$ does not affect the estimation bias, we should select its value so as to minimize the estimation variance. Specifically, the following theorem illustrates the optimal value of $w_1$.

\begin{theorem} \label{thm:ovv}
The variance of second-round estimation - i.e., $\tilde{v}_2 = w_1 \cdot (\tilde{v}_1 + \tilde{v}_\mathrm{c}) + (1-w_1) \cdot \tilde{v}_\mathrm{d}$ is
\begin{align}
w_1^2 \cdot (\frac{\sigma^2_\mathrm{c}}{h_1} + \frac{\sigma^2_1}{h}) + (1 - w_1)^2 \cdot \frac{\sigma^2_\mathrm{d}}{h_2}\label{equ:vrs}
\end{align}
where $\sigma^2_\mathrm{c}$ is the variance of estimation for $|D_2| - |D_1|$ from a single drill down, i.e., the variance of $|q_2(r_i)|/p(q_2(r_i)) - |q_1(r_i)|/p(q_1(r_i))$ for a random $r_i$, while $\sigma^2_1$ and $\sigma^2_d$ are the variance of estimation for $|D_1|$ and $|D_2|$ for a single drill down, respectively. This variance of $\tilde{v}_2$ takes the minimum value when
\begin{align}
w_1 = \frac{\sigma^2_\mathrm{d}/h_2}{\sigma^2_\mathrm{c}/h_1 + \sigma^2_1/h + \sigma^2_\mathrm{d}/h_2}.
\end{align}
\end{theorem}

\begin{proof}
We start the proof by deriving the variance of $\tilde{v}_2$ as follows.
\begin{align}
\sigma^2_2 & = Var (\tilde{v}_2)\\
&= Var (w_1 \cdot (\tilde{v}_1 + \tilde{v}_\mathrm{c}) + (1-w_1) \cdot \tilde{v}_\mathrm{d})\\
&= w_1^2 \cdot Var (\tilde{v}_1 + \tilde{v}_\mathrm{c}) + (1-w_1)^2 \cdot Var (\tilde{v}_\mathrm{d})\\
&= w_1^2 \cdot (\frac{\sigma^2_\mathrm{c}}{h_1} + \frac{\sigma^2_1}{h}) + (1 - w_1)^2 \cdot \frac{\sigma^2_\mathrm{d}}{h_2}\\
&= w_1^2 \cdot (\frac{\sigma^2_\mathrm{c}}{h_1} + \frac{\sigma^2_1}{h} + \frac{\sigma^2_\mathrm{d}}{h_2}) - 2w_1 \cdot \frac{\sigma^2_\mathrm{d}}{h_2} +  \frac{\sigma^2_\mathrm{d}}{h_2}
\end{align}
One can see from the above derivation the correctness of (\ref{equ:vrs}). 

We now derive the optimal value of $\omega_1$ which leads to the minimum value of $\sigma^2_2$. Note that the minimum value of $\sigma^2_2$ is taken when $\mathrm{d} \sigma^2_2 / \mathrm{d} \omega_1 = 0$. Thus, we start by deriving $\mathrm{d} \sigma^2_2 / \mathrm{d} \omega_1$ as follows.
\begin{align}
\frac{\mathrm{d} \sigma^2_2}{\mathrm{d} w_1}  &= 2w_1 \cdot  (\frac{\sigma^2_\mathrm{c}}{h_1} + \frac{\sigma^2_1}{h} + \frac{\sigma^2_\mathrm{d}}{h_2}) - 2 \cdot \frac{\sigma^2_\mathrm{d}}{h_2}
\end{align}
Let $\frac{\mathrm{d} \sigma^2_2}{\mathrm{d} w_1} = 0$,  we have
\begin{align}
w_1 \cdot  (\frac{\sigma^2_\mathrm{c}}{h_1} + \frac{\sigma^2_1}{h} + \frac{\sigma^2_\mathrm{d}}{h_2}) = \frac{\sigma^2_\mathrm{d}}{h_2}
\end{align}
Thus, $\sigma^2_2$ takes the minimum value when
\begin{align}
w_1 = \frac{\sigma^2_\mathrm{d}/h_2}{\sigma^2_\mathrm{c}/h_1 + \sigma^2_1/h + \sigma^2_\mathrm{d}/h_2}.
\end{align}
\end{proof}

One can see that the optimal value of $w_1$ depends on three estimation variances: $\sigma^2_1$, $\sigma^2_\mathrm{d}$, and $\sigma^2_\mathrm{c}$. Obviously, RS-ESTIMATOR cannot directly compute these {\em population variances} because of the lack of knowledge of real $|D_1|$ and $|D_2|$. Nonetheless, we can estimate their values through {\em sample variances} - i.e., the actual variance of estimations produced by the $h$, $h_2$, and $h_1$ drill downs, respectively - and then correct the estimation through Bessel's correction \cite{HA69}.

\vspace{1mm}
\noindent{\bf Distribution of Query Budget:} We now consider how to optimally distribute the query budget for Round $\mathcal{R}_2$ so as to minimize the estimation variance of $\tilde{v}_2$. According to Theorem~\ref{thm:ovv}, when $w_1$ takes the optimal value, the estimation variance of $\tilde{v}_2$ becomes
\begin{align}
\epsilon^2_2 = \frac{(\sigma^2_\mathrm{c}/h_1 + \sigma^2_1/h) \cdot \sigma^2_\mathrm{d}/h_2}{\sigma^2_\mathrm{c}/h_1 + \sigma^2_1/h + \sigma^2_\mathrm{d}/h_2}. \label{equ:vex}
\end{align}

Let $g_\mathrm{c}$ and $g_\mathrm{d}$ be the average query cost per updated and new drill down, respectively. Given a per-round query budget $G$, the following corollary illustrates how to optimally distribute query budget.

\begin{corollary} \label{thm:oqd}
For a given query budget $G$ such that $g_\mathrm{c} \cdot h_1 + g_\mathrm{d} \cdot h_2 = G$, the value of $h_1$ which minimizes $\epsilon^2_2$ is
\begin{small}
\begin{align}
h_1 = \max\left(0, \min\left(\frac{G}{g_\mathrm{c}}, h, \frac{h \cdot (\sqrt{g_\mathrm{d} \cdot \sigma^2_\mathrm{d} \cdot \sigma^2_\mathrm{c} / g_\mathrm{c}} - \sigma^2_\mathrm{c})}{\sigma^2_1}\right)\right). \label{equ:dqb}
\end{align}
\end{small}
\end{corollary}

Once again, the optimal distribution depends on $\sigma^2_1$, $\sigma^2_\mathrm{d}$ and $\sigma^2_\mathrm{c}$, which, as described above, can be approximated using sample variances.  Note that the optimal distribution here also depends on $g_\mathrm{c}$ and $g_\mathrm{d}$, which we also estimate according to the bootstrapping drill downs conducted in the beginning, as shown in the detailed algorithm described in the next subsection.

\vspace{1mm}
\noindent {\bf Comparison with REISSUE-ESTIMATOR:} The comparison between REISSUE- and RS-ESTIMATOR can be directly observed from Corollary~\ref{thm:oqd}, which indicates that (1) when the database undergoes little change, RS-ESTIMATOR mostly conducts new drill downs - leading to a lower estimation error than REISSUE, and (2) when the database changes drastically, RS-ESTIMATOR will be automatically reduced to REISSUE-ESTIMATOR.

To understand why, consider the following examples. When the database is not changed, i.e., $\sigma^2_\mathrm{c} = 0$, Corollary~\ref{thm:oqd} indicates that $h_1 = 0$ - i.e., all query budget will be devoted to initializing new drill downs. On the other hand, when the database undergoes fairly substantial changes -  e.g., when $\sigma^2_\mathrm{c} \approx \sigma^2_\mathrm{d} \approx \sigma^2_1$, Corollary~\ref{thm:oqd} indicates $h_1 = \max(0, \min(G/g_\mathrm{c}, h, h \cdot (\sqrt{g_\mathrm{d}/g_\mathrm{c}} - 1)))$. Recall from \S\ref{sec:rie} that $g_\mathrm{d} > g_\mathrm{c}$ in most practical scenarios. Thus, we have $h_1 = \min(G/g_\mathrm{c}, h)$ - i.e., the query budget is devoted to updating Round-$\mathcal{R}_1$ drill downs as much as possible, exactly like what REISSUE-ESTIMATOR would do in this circumstance.

\subsection{Algorithm RS-ESTIMATOR}

\noindent {\bf Extension to General Aggregates and Multi-Rounds:} We discussed in the last subsection an example of applying RS-ESTIMATOR to estimate COUNT(*) in the second round $\mathcal{R}_2$. We now consider the extension to generic cases, specifically at the following two fronts: (1) in an arbitrary round $\mathcal{R}_i$, by leveraging multiple historic rounds of estimations, and (2) for estimating generic aggregates defined in \S\ref{sec:oae}.

To understand the challenge of these extensions, consider the following scenario: At Round $\mathcal{R}_3$, we detect more changes to the database (i.e., a larger $\sigma^2_\mathrm{c}$), and therefore need a larger number of updated drill downs (i.e., $h_1$) than the second round. One can see what happens here is that there are not enough second-round drill downs for us to update. Thus, we may have to update some of the drill downs from the first round, leading to two problems:
\begin{itemize}
\item Since the difference between first- and third-round results may be larger, we have to adjust the distribution of query budget accordingly - e.g., by giving a higher budget to new drill downs (i.e., larger $h_2$).
\item These cross-round updated drill downs will not be useful for estimating queries such as $|D_3| - |D_2|$, because we do not know their results for Round $\mathcal{R}_2$.
\end{itemize}

Before presenting results that address these two challenges, we first introduce some notations. Specifically, for any aggregate query $Q$ (to be estimated) and any drill down (with signature) $r_i$ which we last updated in Round $\mathcal{R}_x$, we denote the estimation produced by updating $r_i$ at Round $\mathcal{R}_j$ by a function\footnote{\small{Note that we are reusing the symbol $f(\cdot)$ which we used to denote the aggregate estimation function in \S\ref{sec:ffr}. The meaning of $f(\cdot)$ remains the same, yet the input here is not just the current-round result but results from both $\mathcal{R}_x$ and $\mathcal{R}_j$. In addition, the function now depends on the aggregate function $Q$ to be estimated.}} $f_Q(x, q_j(r_i))$, where $q_j(r_i)$ is the top non-overflowing query in the path with signature $r_i$ at Round $\mathcal{R}_j$. Note that $f_Q(\cdot)$ differs for different aggregate query $Q$. For example,
\begin{itemize}
\item If $Q$ is SELECT SUM($A_1$) FROM D, we have
\begin{align}
f_Q(x, q_j(r_i)) = \tilde{Q}_x + \sum_{t \in q_j(r_i)} \frac{t[A_1]}{p(q_j(r_i))} \hspace{1mm} - \sum_{t \in q_x(r_i)} \frac{t[A_1]}{p(q_x(r_i))}, \nonumber
\end{align}
where $\tilde{Q}_x$ is the estimation we produced for $Q$ in Round $\mathcal{R}_x$. One can see that what we discussed (for $Q = |D|$) in \S\ref{sec:kir} is a special case of this result (when $x = 1$, $j = 2$, $t[A_i] \equiv 1$ and $\tilde{Q}_1 = |D_1|$.)
\item If $Q$ is $|D_j| - |D_{j-1}|$ and $x = j-1$, we have $f_Q(x, q_j(r_i)) = |q_j(r_i)| / p(q_j(r_i)) - |q_x(r_i)|/p(q_x(r_i))$.
\item If $Q$ is $|D_j| - |D_{j-1}|$ and $x < j-1$, we have $f_Q(x, q_j(r_i)) = |q_j(r_i)| / p(q_j(r_i)) - |\tilde{D}_{j-1}|$ where $|\tilde{D}_{j-1}|$ is our estimation of $|D_{j-1}|$ from the previous round. In this case, no result from Round $\mathcal{R}_x$ (e.g., $q_x(r_i)$) is used in the computation of $f_Q$ - i.e., updating old and initiating new drill downs become equivalent.
\end{itemize}

Recall from the discussion in \S\ref{sec:are} and Theorem~\ref{thm:uep} that, for any COUNT and SUM aggregate query $Q$ defined in \S\ref{sec:oae}, we can find $f_Q(\cdot)$ that produces an unbiased estimation of $Q$. Given $f_Q(\cdot)$, we are now ready to theoretically analyze the optimal distribution of query budget for estimating generic aggregates in multi-round scenarios. First, we have a corollary to Theorem~\ref{thm:ovv}.

\begin{corollary} \label{cor:ge1}
At Round $\mathcal{R}_j$, if RS-ESTIMATOR updates $c_1, \ldots, c_{j-1}$ drill downs that were last updated in Rounds $\mathcal{R}_1$, $\ldots$, $\mathcal{R}_{j-1}$, respectively, and initiates $c_j$ new drill downs, the optimal estimation for an aggregate $Q$ is
\begin{align}
\sum^j_{x=1} \left(\frac{1}{c_x \cdot v^2_x(c_x) \cdot \sum^j_{y=1} \frac{1}{v^2_y(c_y)} } \cdot \sum_{i \in [1, c_x]} f_Q(i, q_j(r^x_i))\right)
\end{align}
where $r^x_i$ is the $i$-th drill down among the $c_x$ ones that were last updated in Round $\mathcal{R}_x$, $f_Q(j, q_j(r^x_j))$ is the estimation from a new drill down $r^x_j$ in Round $\mathcal{R}_j$, and $v^2_x(c_x)$ is the variance of
\begin{align}
\frac{1}{c_x} \sum^{c_x}_{i=1} f_Q(i, q_j(r^x_i))
\end{align}
taken over the randomness of $r^x_i$. The estimation variance is
\begin{align}
\epsilon^2_j(Q) = \frac{1}{\sum^j_{x=1} (1/v^2_x(c_x))}
\end{align}
\end{corollary}

One can see that Theorem~\ref{thm:ovv} and (\ref{equ:vex}) indeed represent a special case of this corollary when $j = 2$, $c_1 = h_1$, $c_2 = h_2$, $v^2_1(c_1) = \sigma^2_\mathrm{c}/h_1 + \sigma^2_1/h$, and $v^2_2(c_2) = \sigma^2_\mathrm{d}/h_2$. According to Corollary~\ref{cor:ge1}, we can derive the optimal query budget distribution similar to Corollary~\ref{thm:oqd}. This derivation uses the following two key properties of $f_Q(x, q_j(r^x_j))$:
\begin{align}
\beta_x &= \lim_{c_x \to \infty} \mathrm{var}(f_Q(x, q_j(r^x_j))) = \lim_{c_x \to \infty} v^2_x(c_x)\\
\alpha_x &= \lim_{c_x \to \infty} c_x \cdot (\mathrm{var}(f_Q(x, q_j(r^x_j))) - \beta_x)\\
&= \lim_{c_x \to \infty} c_x \cdot (v^2_x(c_x) - \beta_x)
\end{align}
where $\mathrm{var}(\cdot)$ represents the estimation variance.
\begin{corollary} \label{cor:ge2}
For a given query budget $G$ such that $\sum^j_{x=1} g_x \cdot c_x = G$, the value of $c_x$ which minimizes $\epsilon^2_j(Q)$ is
\begin{align}
c_x = \frac{G \cdot \sqrt{g_x/\alpha_x}}{\beta_x \cdot \sum^j_{i = 1}((\sqrt{g_i/\alpha_i} - \alpha_i) \cdot (g_i/\beta_i))} - \frac{\alpha_x}{\beta_x}.
\end{align}
when $\beta_x > 0$ for all $x$. When $\beta_x = 0$ for all $x$, then $c_x = G/g_x$ if $x$ minimizes $\alpha_x \cdot g_x$ and $0$ otherwise (break tie arbitrarily). When $\beta_x > 0$ for some $x$, consider $y$ such that $\beta_y = 0$ and
\begin{align}
\alpha_y \cdot g_y = \min_{i \in [1, j], \beta_i = 0} \alpha_i \cdot g_i.
\end{align}
The optimal distribution is to have
\begin{align}
c_x = \max\left(0, \min\left(\frac{G}{g_x}, \frac{\sqrt{\alpha_x \cdot \alpha_y \cdot g_y / g_x} - \alpha_x}{\beta_x}\right)\right). \label{equ:fod}
\end{align}
\end{corollary}
One can see that Corollary~\ref{thm:oqd} is indeed a special case of this corollary when the last scenario occurs - i.e., one $\beta_2$ (for the $h_2$ new drill downs) is 0, while the other $\beta_1$ (for the $h_1$ updated ones) is not ($\beta = \sigma^2_1/h$). Note that in this special case, we have $\alpha_1 = \sigma^2_\mathrm{c}$ and $\alpha_2 = \sigma^2_\mathrm{d}$ - making (\ref{equ:fod}) equivalent with (\ref{equ:dqb}).

\vspace{1mm}
\noindent {\bf Algorithm Description:} We are ready to describe the generic RS-ESTIMATOR algorithm. At each round $\mathcal{R}_j$, our algorithm starts by conducting a number of {\em bootstrapping} drill downs (line 4) to estimate $v^2_x(c_x)$ and $g_x$ - i.e., the variance of estimation and the query cost for updating a drill down that was last updated at Round $\mathcal{R}_x$ (recall that if $x = j$, such an ``update'' is indeed initializing a new drill down). Specifically, for each $x \in [1, j]$, we conduct $\varpi$ bootstrapping drill downs, where $\varpi$ is a user-specified parameter.

After conducting the $r \cdot j$ drill downs, we obtain the estimations for $\alpha_x$, $\beta_x$ and $g_x$, and compute according to Corollary~\ref{cor:ge2} the optimal distribution of query cost - i.e., $c_x$ drill downs for each $x \in [1, j]$ (line 5). Suppose that there are $h_x$ drill downs that were last updated at Round $\mathcal{R}_x$. For each $x \in [1, j]$, we choose $\min(c_x, h_x)$ drill downs uniformly at random from the $h_x$ ones\footnote{\small{Note that when $x = j$, we simply choose $c_x$ drill downs uniformly at random.}}, and place all chosen drill downs into a pool $\Omega$ (line 6). Then, we conduct drill downs in $\Omega$ according to a random order (line 8) until exhausting all query budget. Note that the reason for doing so is to accommodate the error of estimating $g_x$ - i.e., we might not have enough query budget to completely conduct all $c_1 + \cdots + c_j$ drill downs. Finally, we produce our final estimation of $Q$ according to Corollary~\ref{cor:ge1} (line 9).

Algorithm~\ref{alg:rsEst} provides a pseudocode for RS-ESTIMATOR.

\begin{algorithm}[!htb]
\caption{{\bf RS-ESTIMATOR}}
\begin{algorithmic}[1]
\label{alg:rsEst}
\STATE {\bf Input:} Aggregate query $Q$, round index $j$, bootstrapping query limit per round $\varpi$
\STATE Set drill down pool $\Omega$ = \{\}
\FOR {x=1 \TO j}
\STATE Execute $\varpi$ pilot drill-downs corresponding to $\mathcal{R}_x$
\STATE Estimate update budget $c_x$ using Corollary~\ref{cor:ge2}
\STATE $\Omega = \Omega$ $ \cup $ randomly picked $c_x$ queries from round $\mathcal{R}_x$
\ENDFOR
\STATE Issue randomly chosen drill-downs from $\Omega$ till exhausting query budget
\STATE Produce aggregate estimation according to Corollary~\ref{cor:ge1}.
\end{algorithmic}
\vspace{-1mm}
\end{algorithm}


%
%
%
%

\section{System Design} \label{sec:sim}

While there exists a vast diversity in the properties of the dynamic hidden web databases in the real world,
their differences can be abstracted into two dimensions. The first is the {\em query model} that describes the various ways in which the query (or queries) for aggregate estimation are specified. The second corresponds to the {\em update model} that describes how the database is updated over time. In \S\ref{sec:rie} and \ref{sec:sam}, we assumed a simple objective of accurately estimating one {\em pre-defined} aggregate over dynamic data, and a simple update model where all updates to a hidden database occur at the beginning of a {\em round}. In this section, we describe how to adapt our algorithms for other commonly used query and update models.

\subsection{Dimension 1: Query Model}
The query model over dynamic hidden web databases has some parallels to queries over continuous data streams \cite{BBDMW02}.
In this subsection, we describe the two categories of aggregate estimation queries that could satisfy the need of most users.
\begin{itemize}
\item {\em Stream query model}:
This is also called a {\em pre-defined query model} and
supports aggregate queries that are specified upfront before any query processing over the hidden database starts.
These queries are evaluated continuously as the database evolves and should be constantly updated with the most recent estimation.
\item {\em Ad hoc query model}:
This model supports aggregate queries that could be issued at any point {\em after} our system started tracking the database.
Such queries could be one-time, where the query is evaluated once over, or continuous, where the query is evaluated over
the database instance at the present or some time in the past (albeit after the time our system starts tracking the database).
\end{itemize}

Our algorithms can handle both the query models with few modifications.
Recall that at any point in time, our system maintains a number of drill downs with (past) results from previous rounds
that are refreshed (restarted or reissued) at each round.
For the streaming model, the aggregate estimates constructed from the past drill downs are revised
based on the updated results.
For the ad hoc model, since all tuples retrieved by the previous drill downs can be preserved, one can ``simulate'' the aggregate estimation as if the query was issued prior to the drill downs being done.

Nonetheless, it is important to note that there may still be significant performance differences between aggregate estimation over the two query models.
For example, recall from \S\ref{sec:are} that, with the stream model, if we know all aggregates to process have a particular predicate (e.g., $A_1 = 1$), then we could build a query tree (for drill-down) with all nodes having that predicate (e.g., if the predicate is $A_1 = 1$, then we are essentially taking the subtree of the original query tree under the branch $A_1 = 1$).
This smaller query tree could lead to a more accurate aggregate estimation - a benefit we can no longer enjoy for accommodating the ad hoc query pool. As another example, at each round RS-ESTIMATOR makes decisions on the distribution of query budget according to the aggregate query being processed. Once again, we can no longer reach the optimal distribution when the ad hoc query model is in place.

While the previous sections tracked only a single query, our system could handle multiple queries simultaneously.
However, performing adaptive query processing for multiple queries or
deciding the distribution of drill-downs to update based on the complex correlations between queries is a major challenge.
This is a major challenge even for traditional database query processing and a significant body of prior work exists (see \cite{Deshpande:2007:AQP:1331939.1331940} for a survey).
Adapting these techniques for hidden web databases is non-trivial and is beyond the scope of our paper.

\subsection{Dimension 2: Update Model} \label{sec:dum}
The update model abstracts the frequency of database updates and the query budget.
Depending on whether the database updates happen in batch or in an even fashion, we can categorize dynamic hidden web databases into two types. One is the {\em round-update model} discussed so far in the paper. The other is {\em constant-update model} - i.e., the database may be updated at any time, even in the middle of the execution of our algorithms.


Adapting our (or any) algorithms for constant update model is much trickier.
The notion of rounds cannot be adapted to handle such a case as the database is updated at any point in time,
even in the middle of algorithm execution.
A even major problem is the definition of the ground truth -
how to define the aggregate over such an unpredictably changing database.
Should we try to estimate the aggregate for the database instance that existed when the algorithm started? Or is it the end?
While both definitions sound fine, it is easy to see no algorithm can guarantee unbiased estimation/sampling
no matter how the ground truth is defined.
Nevertheless, our algorithm can be easily adapted to deal with this model
(albeit without the unbiasedness guarantees that may otherwise be provided).
A natural way is to utilize the concept of rounds by ignoring all updates after a specific point in time.
Recall from \S\ref{sec:sam} that RS-ESTIMATOR uses a set of bootstrapping queries to estimate the amount of change.
After this pilot phase is over, we ignore further changes and utilize the information obtained from the pilot phase
for deciding on how to redistribute queries or how to recompute new estimation.
Notice that this estimate could be biased by the tuples that were inserted before the end of pilot phase.

\noindent {\bf Other Issues:} 
While our data model makes a simplistic no-NULL assumption, our algorithms readily handle NULL values as long as the hidden database has a clear policy on how to handle NULL in processing search queries.
For example, if the database allows one to include predicate $A_i$ IS NULL in a search query, then NULL can simply be treated as yet another value in the attribute domain. If, on the other hand, the database features a broad match policy and returns a tuple $t$ with $A_i$ being NULL to queries with any $A_i = v_{ij}$, then we can still precisely compute the probability for $t$ to be returned (by simply multiplying it with the domain size of $A_i$) and use REISSUE- or RS-ESTIMATOR as is.

\vspace{2mm}
\section{Experimental Evaluation}
\vspace{-2mm}
\begin{figure*}[ht]
\begin{minipage}[t]{0.23\linewidth}
\centering
\includegraphics[width = 40mm]{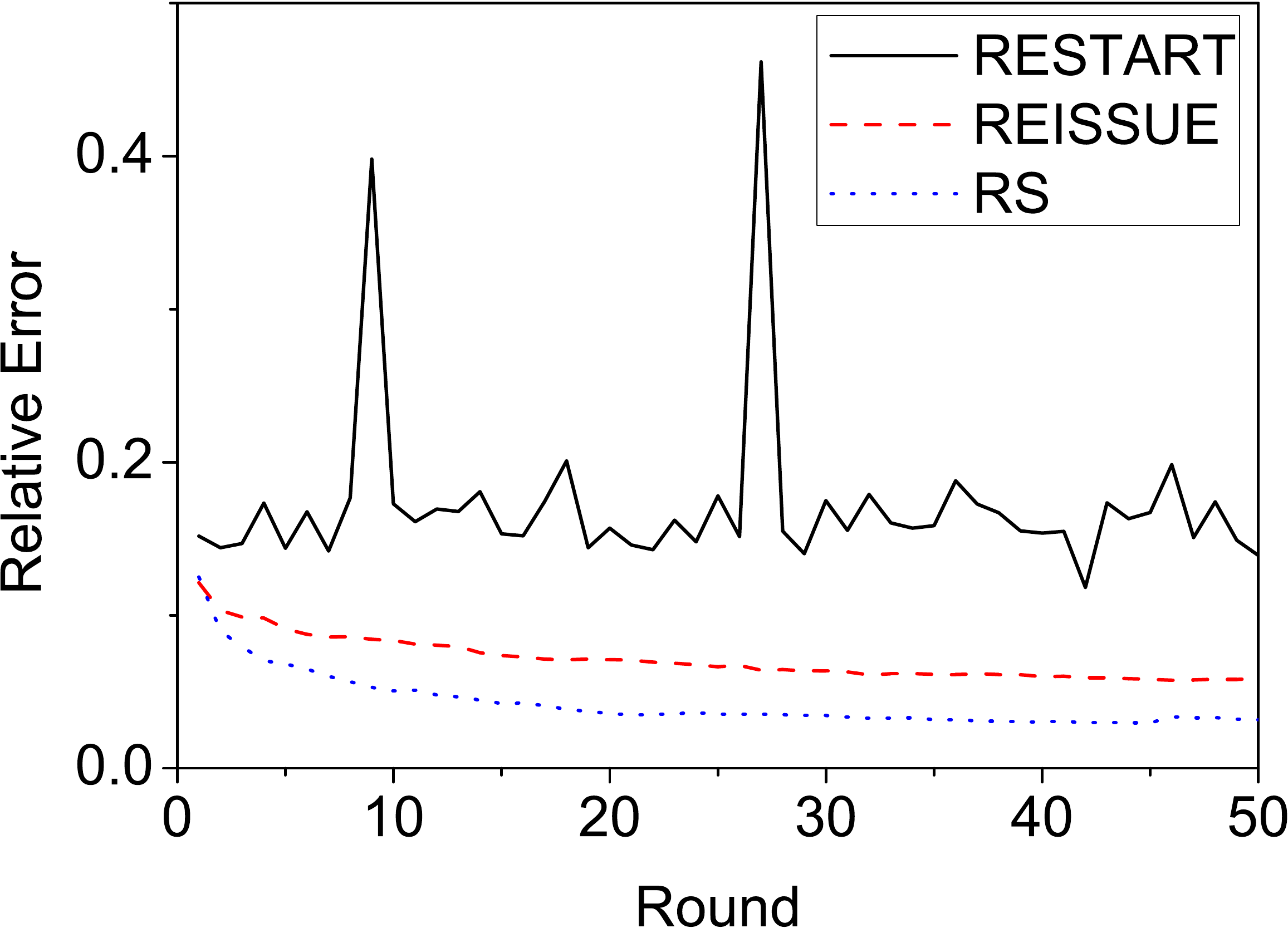}
\vspace{-2mm}\caption{Relative Error}
\label{fig1}
\end{minipage}
\hspace{1mm}
\begin{minipage}[t]{0.23\linewidth}
\centering
\includegraphics[width = 40mm]{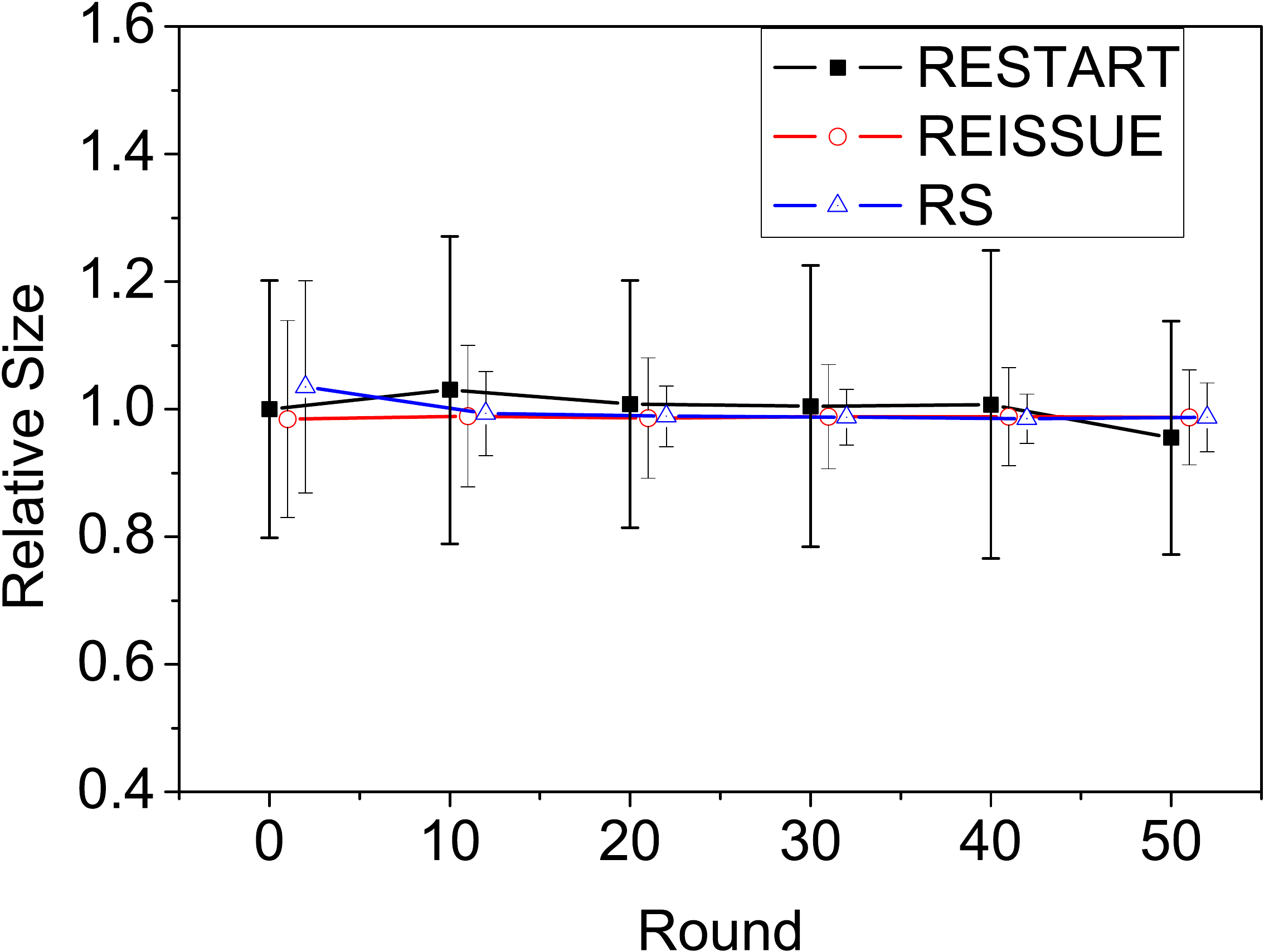}
\vspace{-2mm}\caption{Error Bar}
\label{fig2}
\end{minipage}
\hspace{1mm}
\begin{minipage}[t]{0.23\linewidth}
\centering
\includegraphics[width = 40mm]{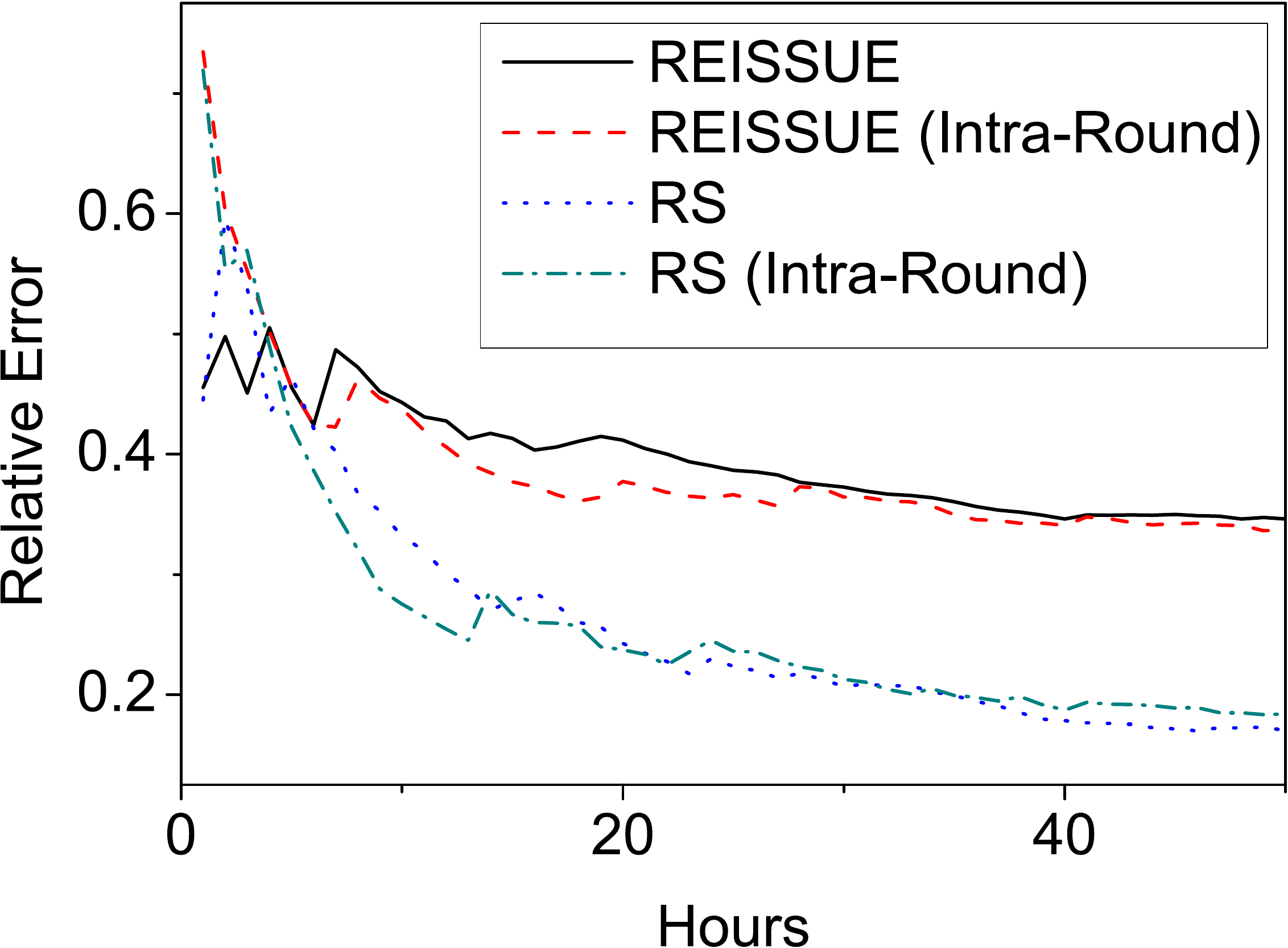}
\vspace{-2mm}\caption{Intra-Round}
\label{fig3}
\end{minipage}
\hspace{3mm}
\begin{minipage}[t]{0.23\linewidth}
\centering
\includegraphics[width = 40mm]{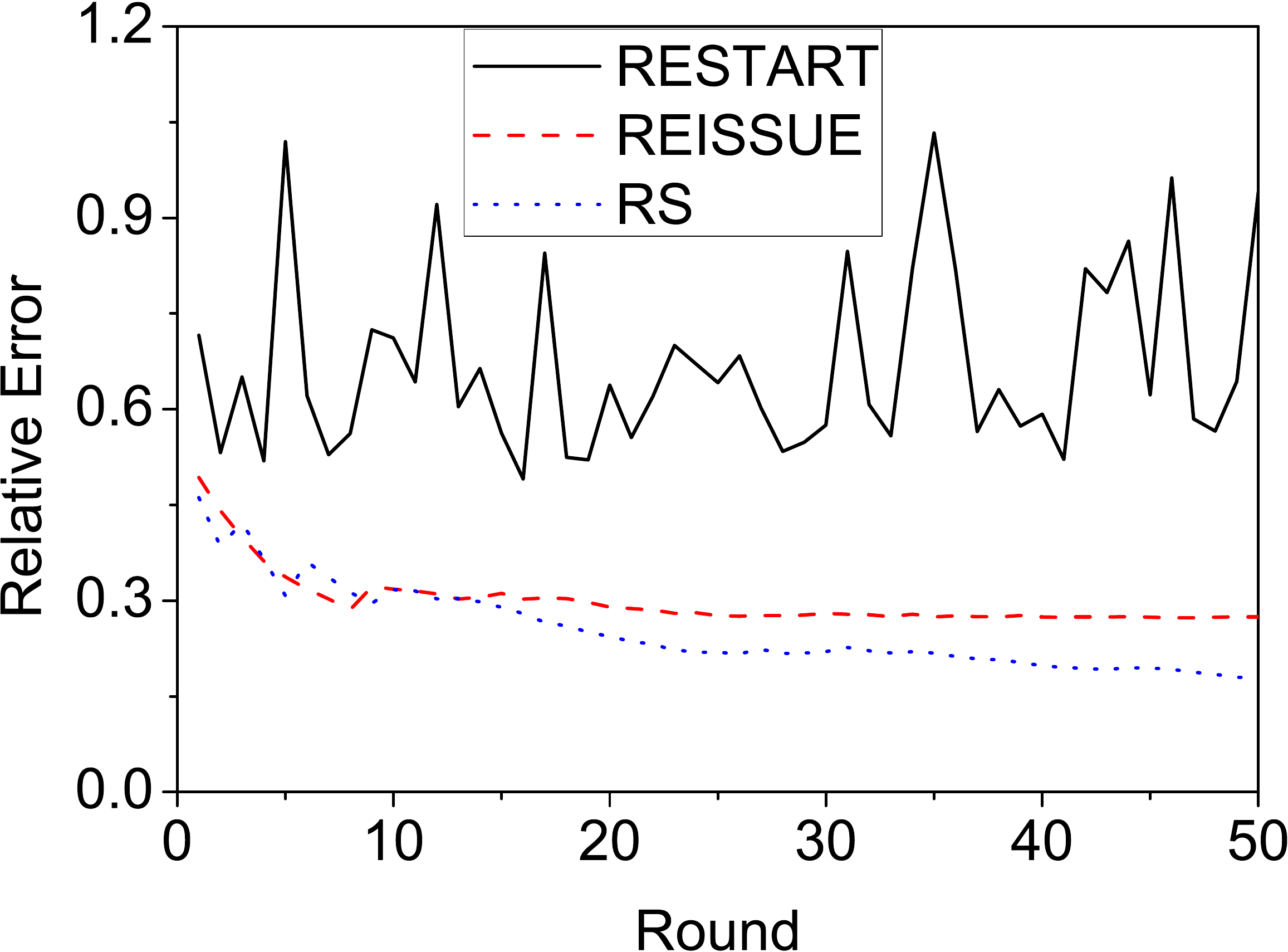}
\vspace{-2mm}\caption{Little Change}
\label{fig4}
\end{minipage}
\hspace{-2mm}
\end{figure*}

\begin{figure*}[t]
\begin{minipage}[t]{0.22\linewidth}
\centering
\includegraphics[width = 40mm]{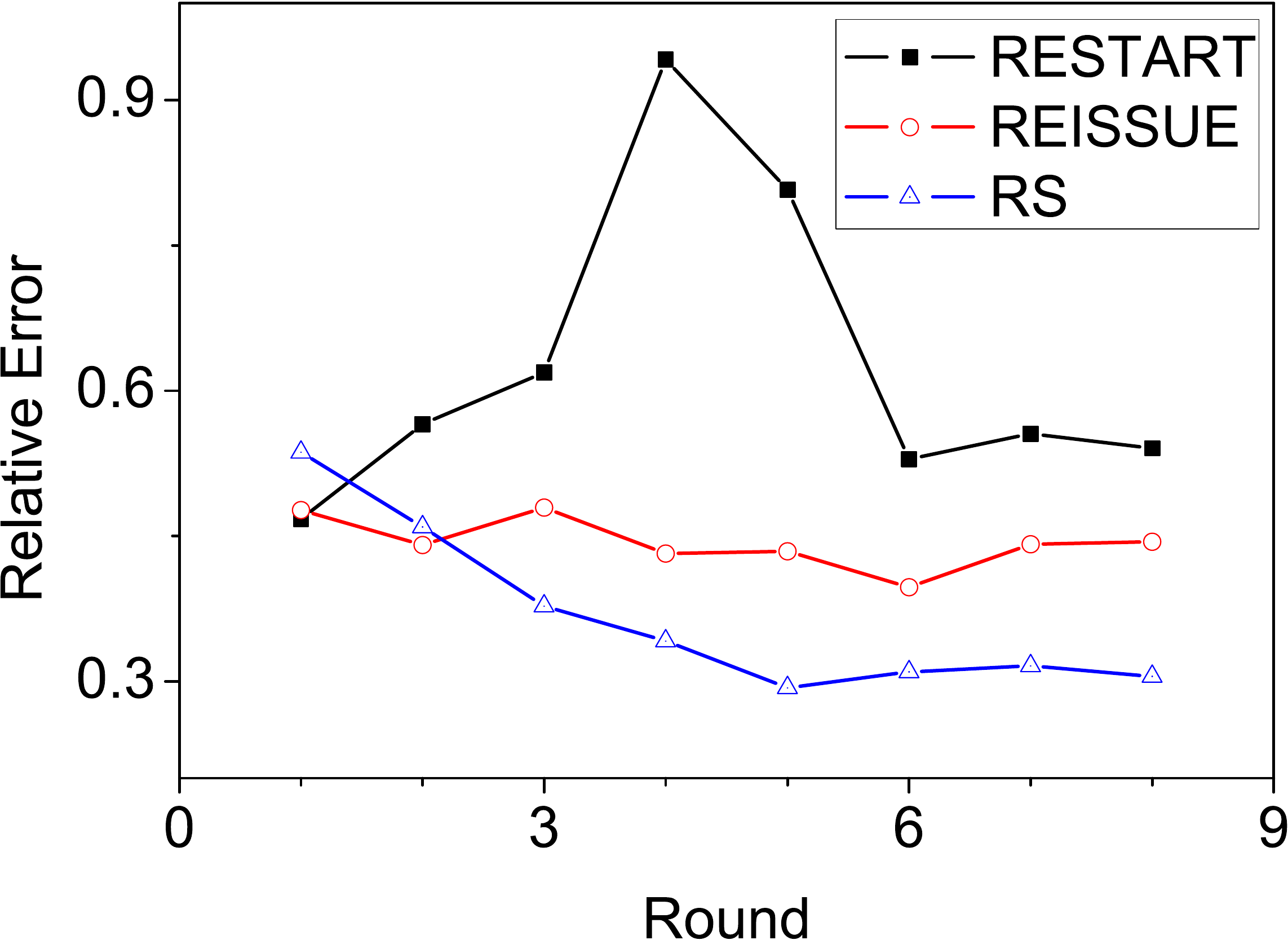}
\vspace{-5.5mm}\caption{Big Change}
\label{fig5}
\end{minipage}
\begin{minipage}[t]{0.23\linewidth}
\centering
\includegraphics[width = 40mm]{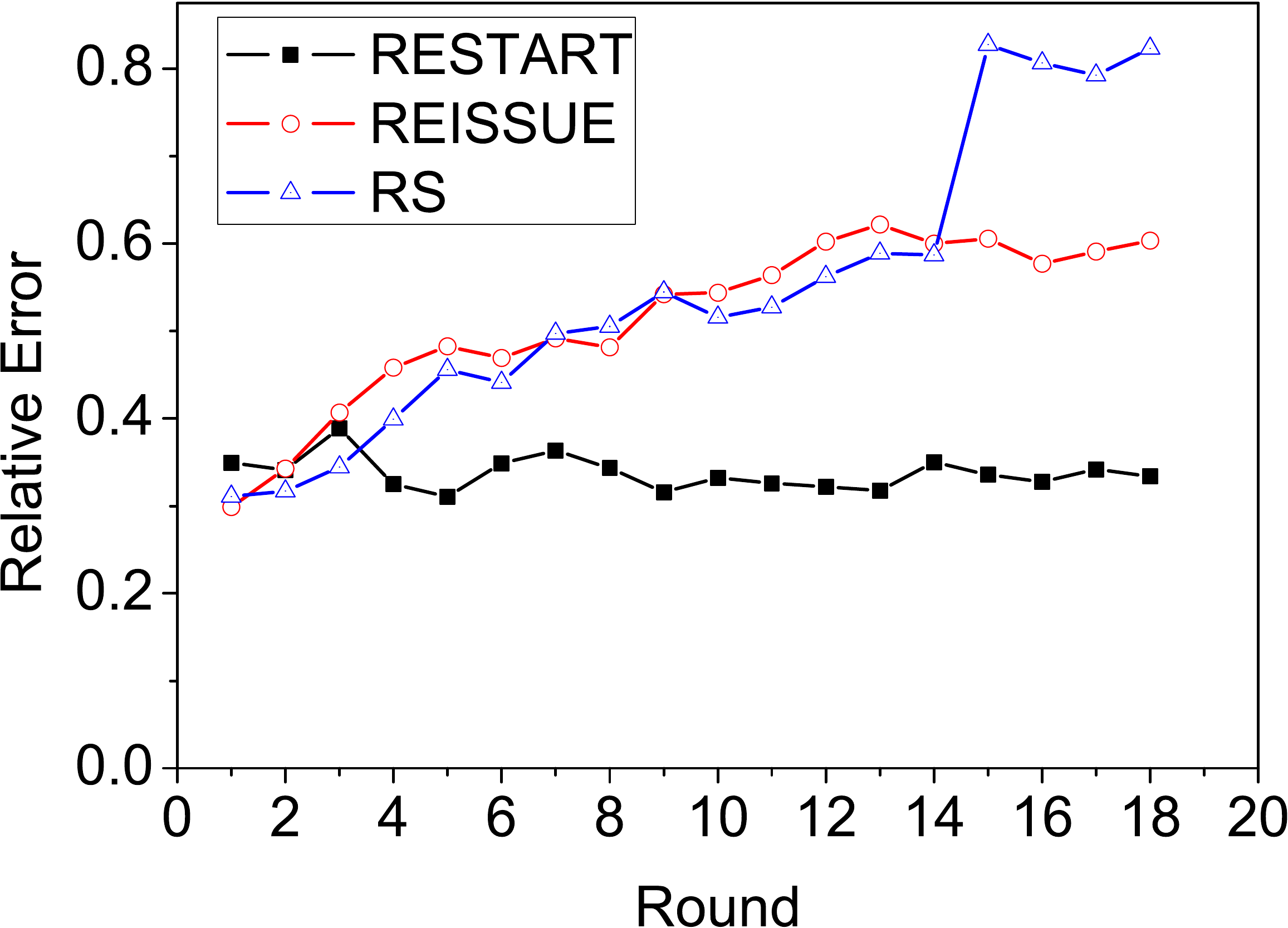}
\vspace{-2mm}\caption{Big Change ($k = 1$)}
\label{fig6}
\end{minipage}
\hspace{3mm}
\begin{minipage}[t]{0.23\linewidth}
\centering
\includegraphics[width = 40mm]{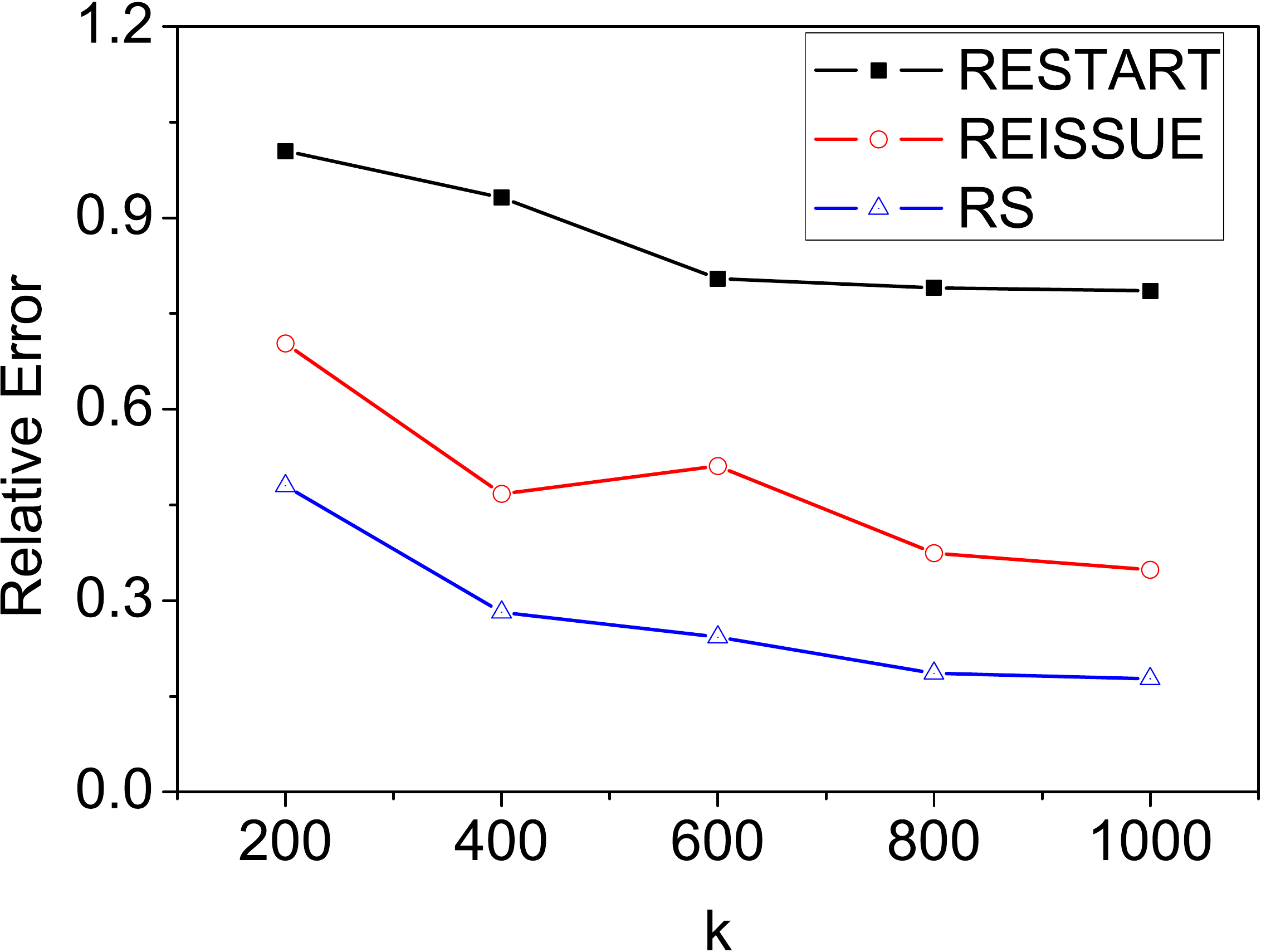}
\vspace{-2mm}\caption{Effect of $k$}
\label{fig7}
\end{minipage}
\hspace{3mm}
\begin{minipage}[t]{0.22\linewidth}
\centering
\includegraphics[width = 40mm]{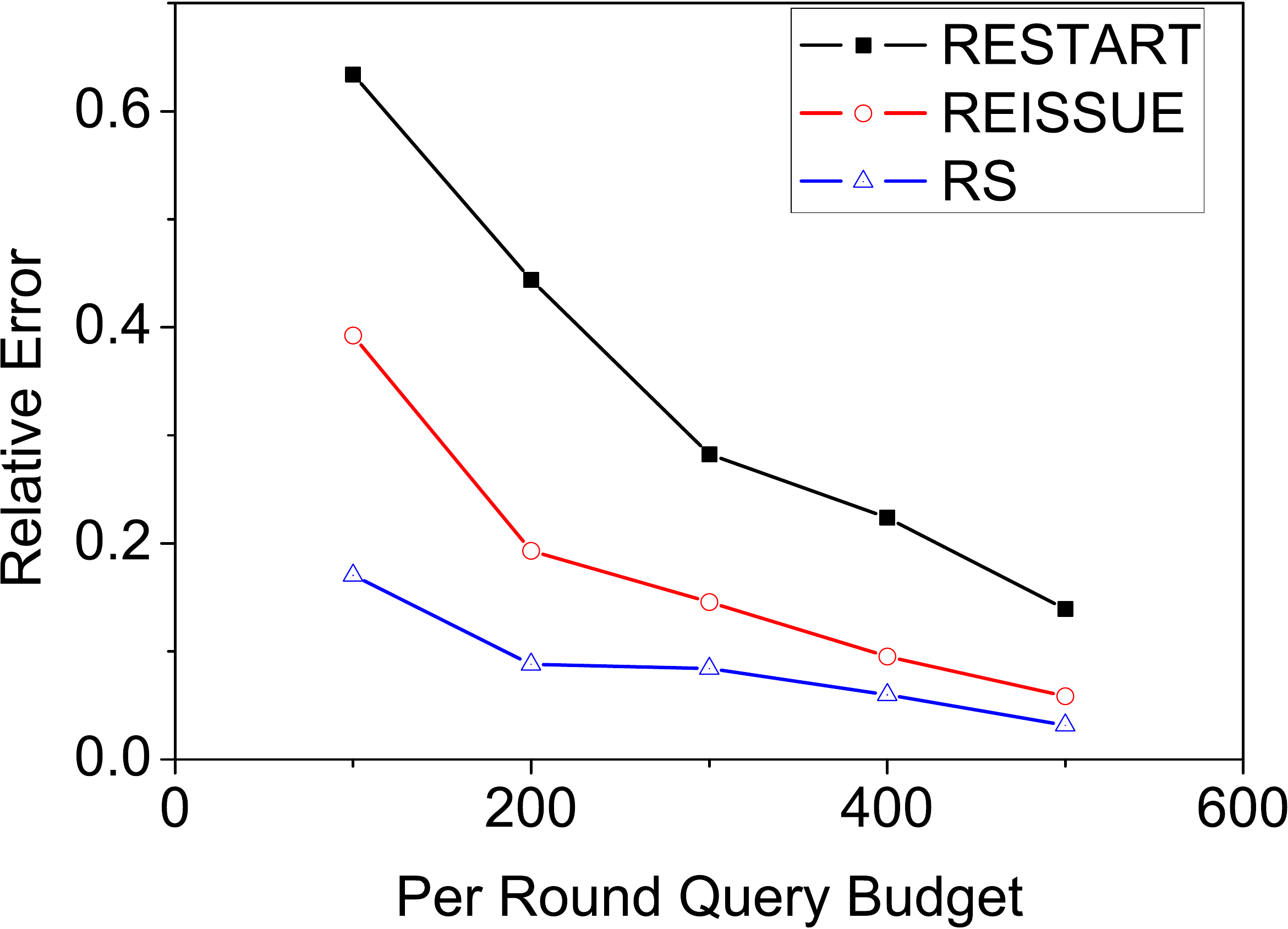}
\vspace{-5.5mm}\caption{Query Budget}
\label{fig8}
\end{minipage}
\hspace{-2mm}
\end{figure*}

\subsection{Experimental Setup}
\noindent{\bf Datasets:} We tested our algorithm over three datasets: (1) a real-world (categorical) web database Yahoo!~Autos to which we have full offline access, (2) a live experiment on Amazon.com, and (3) a live experiment on eBay.com.

Yahoo!~Autos dataset \cite{DJJ+10} is a snapshot of the Yahoo!~Auto database and contains 188,917 unique tuples and 38 attributes, with attribute domain sizes ranging from 2 to 38. The default insertion/deletion schedule we used is to start with 170,000 (chosen uniformly at random without replacement) tuples at the beginning and, for each round, insert 300 randomly selected tuples not currently in the database, and delete 0.1\% of the existing tuples.  The default search interface is top-1000, and the default query budget per round is $G = 100$.
Since we had complete access to the dataset, we were able to build locally an exact implementation of the top-$k$ web interface and precisely compute the aggregate query answers and, thereby, the estimation error generated by our algorithms.

We shall describe detailed settings for the live experiments at the end of this section.

\vspace{1mm}
\noindent{\bf Algorithms Evaluated:} We tested three algorithms discussed in the paper: the baseline RESTART-ESTIMATOR (repeatedly execute \cite{DJJ+10} for each round) and our REISSUE-ESTIMATOR and RS-ESTIMATOR. We abbreviate their names as RESTART, REISSUE, and RS in most of the section. RESTART and REISSUE are parameter-free algorithms (other than the database-controlled parameter of query budget per round $G$). RS has one parameter $\varpi$ which is the number of bootstrapping drill downs performed for each historic round. We set the default setting to $\varpi = 10$ in the experiments.

\vspace{1mm}
\noindent{\bf Performance Measures:} For query cost, we focused on the number of queries issued through the web interface of the hidden database. For estimation accuracy, we measured the relative error (i.e., $|\tilde{\theta} - \theta| / |\theta|$ for an estimator $\tilde{\theta}$ of aggregate $\theta$), and also plotted the error bars for raw estimations.

\begin{figure*}[ht]
\begin{minipage}[t]{0.23\linewidth}
\centering
\includegraphics[width = 40mm]{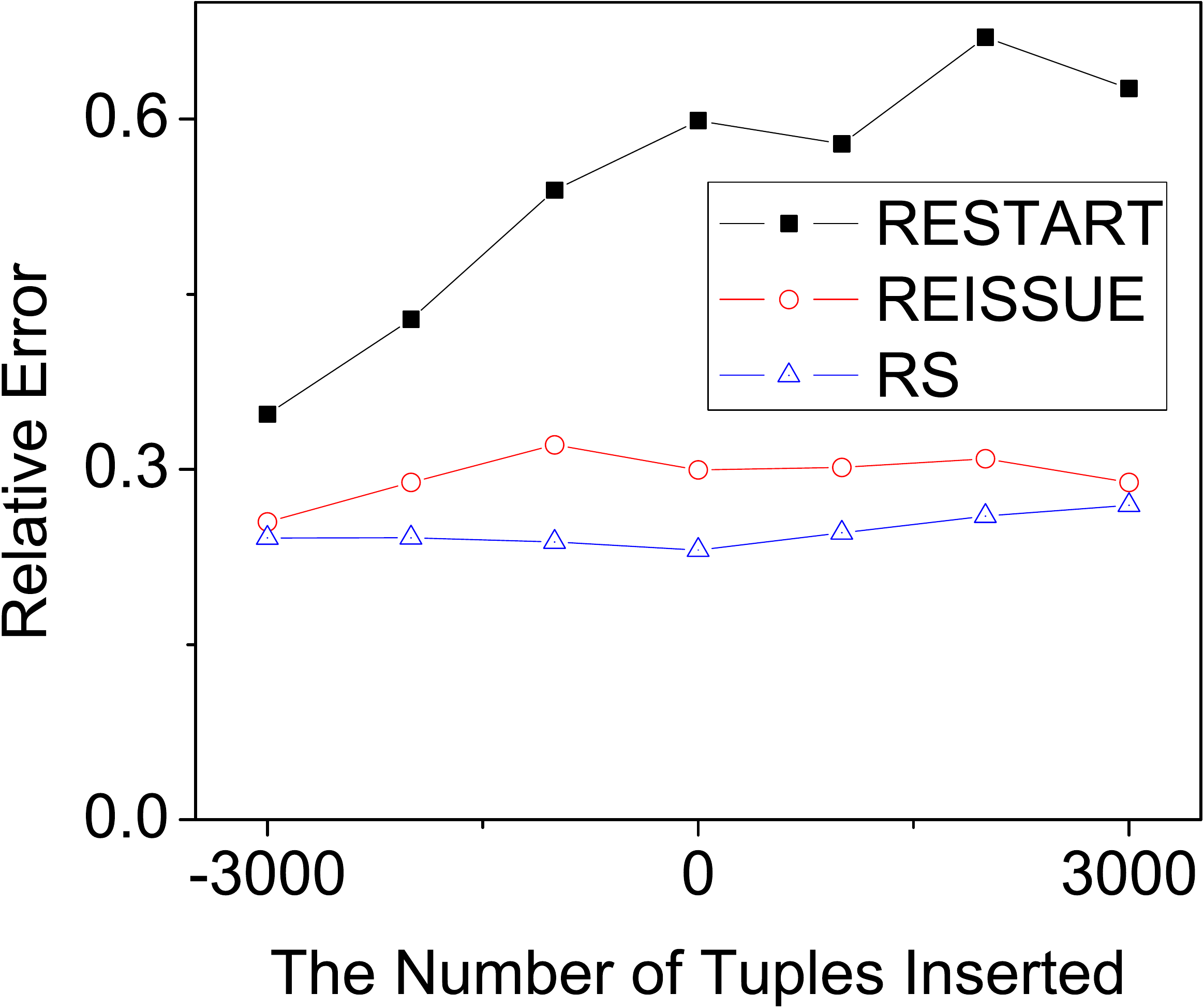}
\vspace{-2mm}\caption{Insertion/Deletion}
\label{fig9}
\end{minipage}
\hspace{3mm}
\begin{minipage}[t]{0.23\linewidth}
\centering
\includegraphics[width = 40mm]{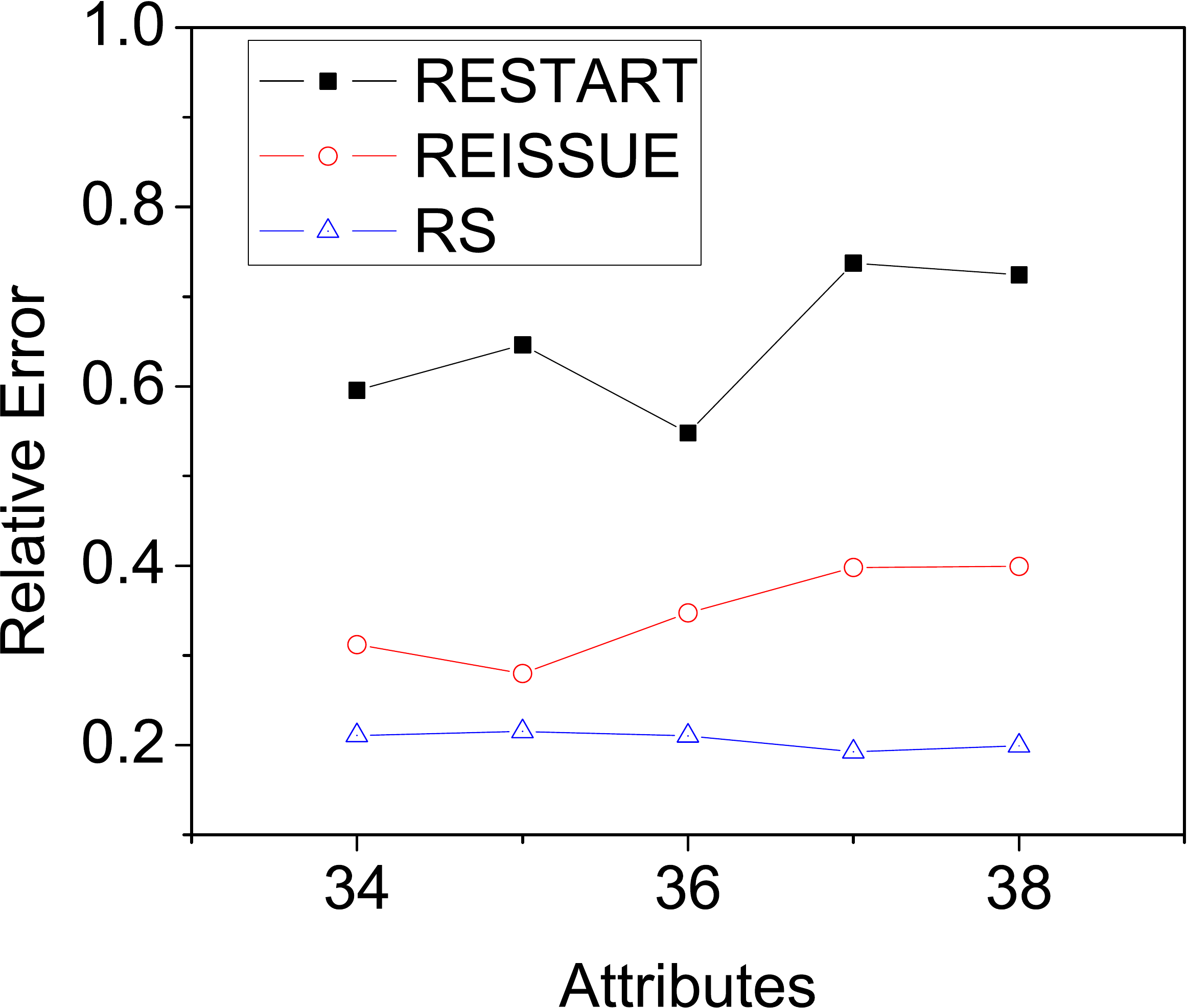}
\vspace{-2mm}\caption{Effect of $m$}
\label{fig10}
\end{minipage}
\hspace{3mm}
\begin{minipage}[t]{0.22\linewidth}
\centering
\includegraphics[width = 40mm]{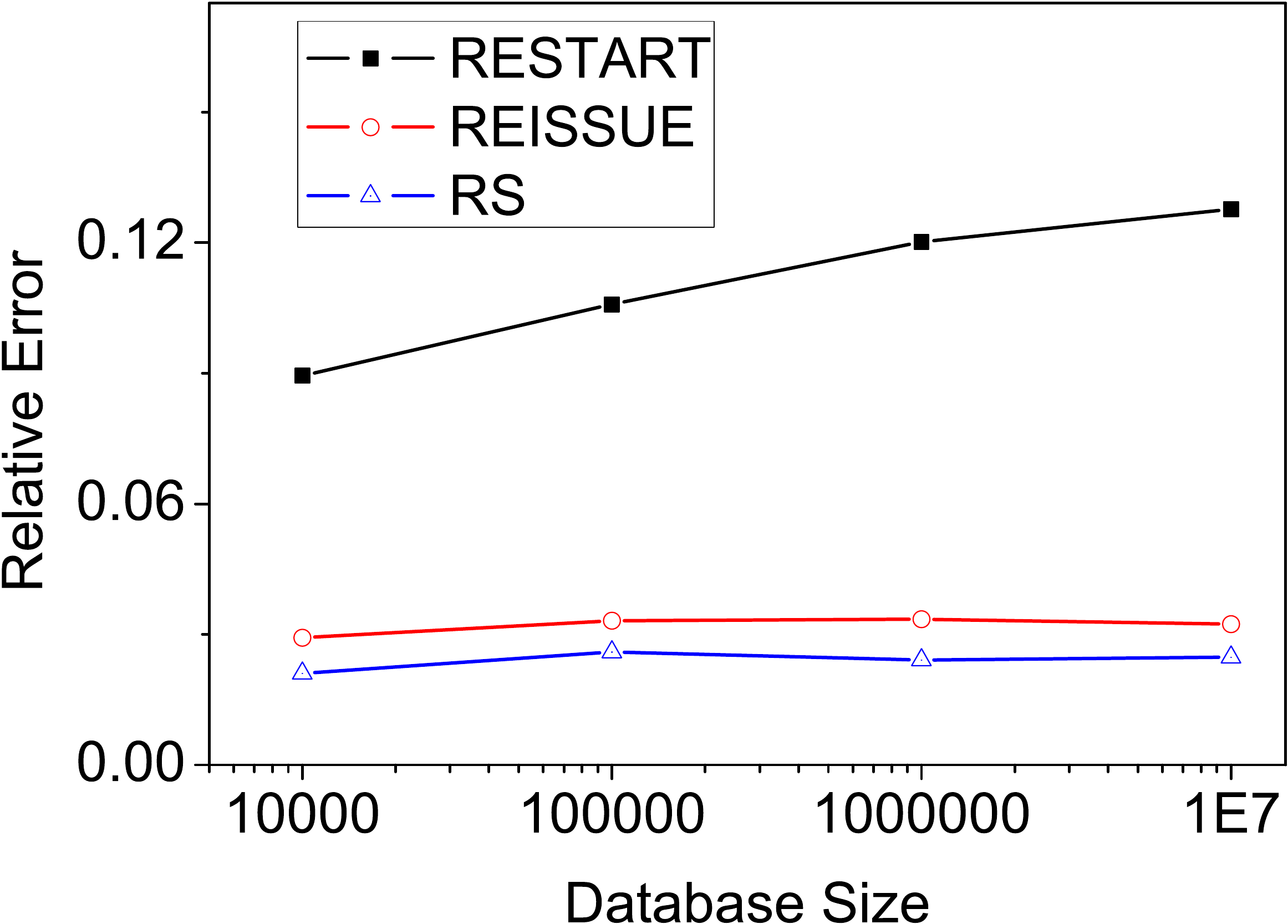}
\vspace{-5.5mm}\caption{Effect of $|D_1|$}
\label{fig11}
\end{minipage}
\hspace{3mm}
\begin{minipage}[t]{0.22\linewidth}
\centering
\includegraphics[width = 40mm]{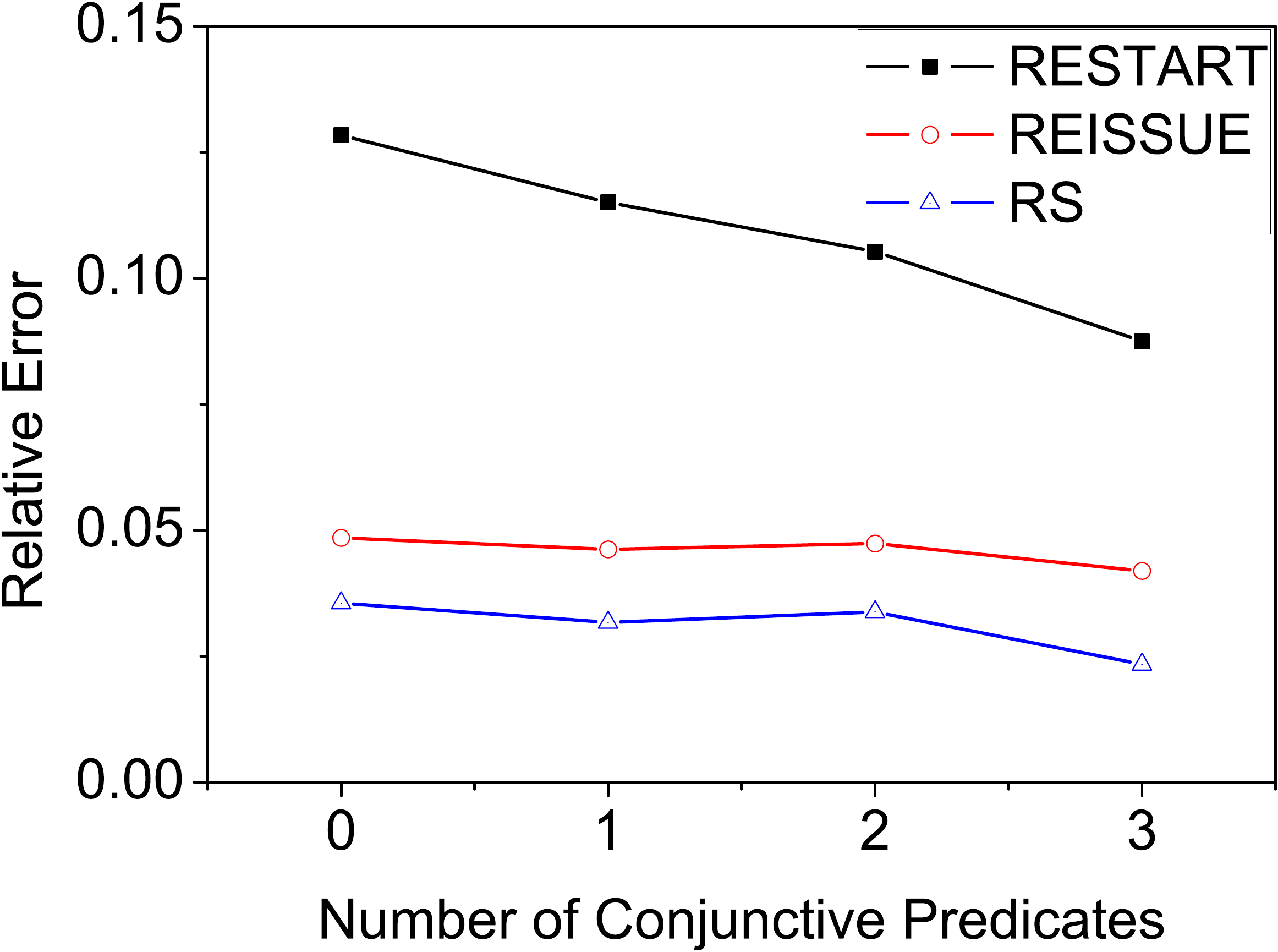}
\vspace{-5.5mm}\caption{SUM w/ Conditions}
\label{fig12}
\end{minipage}
\hspace{-2mm}
\end{figure*}

\begin{figure*}[ht]
\begin{minipage}[t]{0.23\linewidth}
\centering
\includegraphics[width = 40mm]{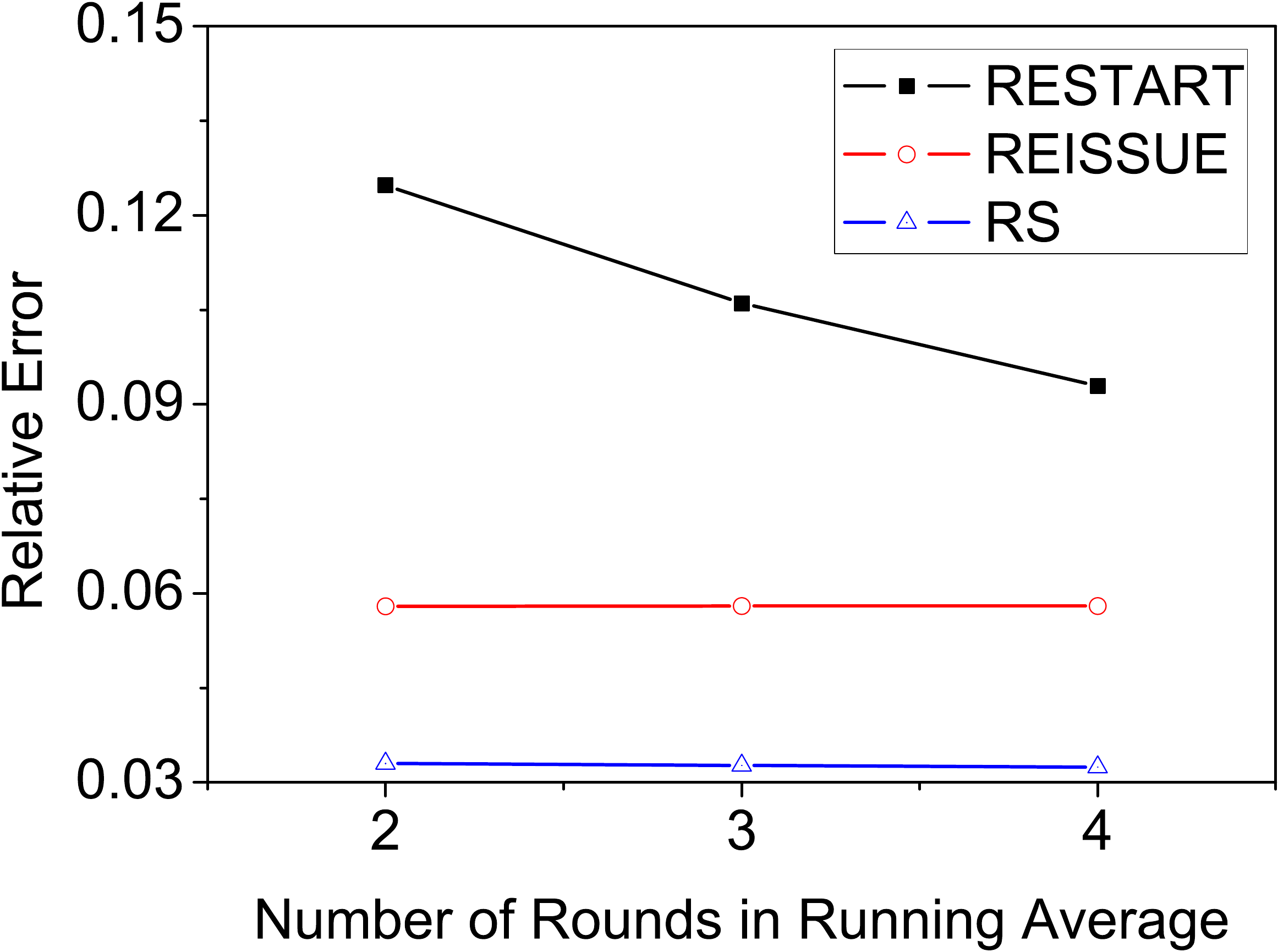}
\vspace{-2mm}\caption{Running Average}
\label{fig13}
\end{minipage}
\hspace{3mm}
\begin{minipage}[t]{0.23\linewidth}
\centering
\includegraphics[width = 40mm]{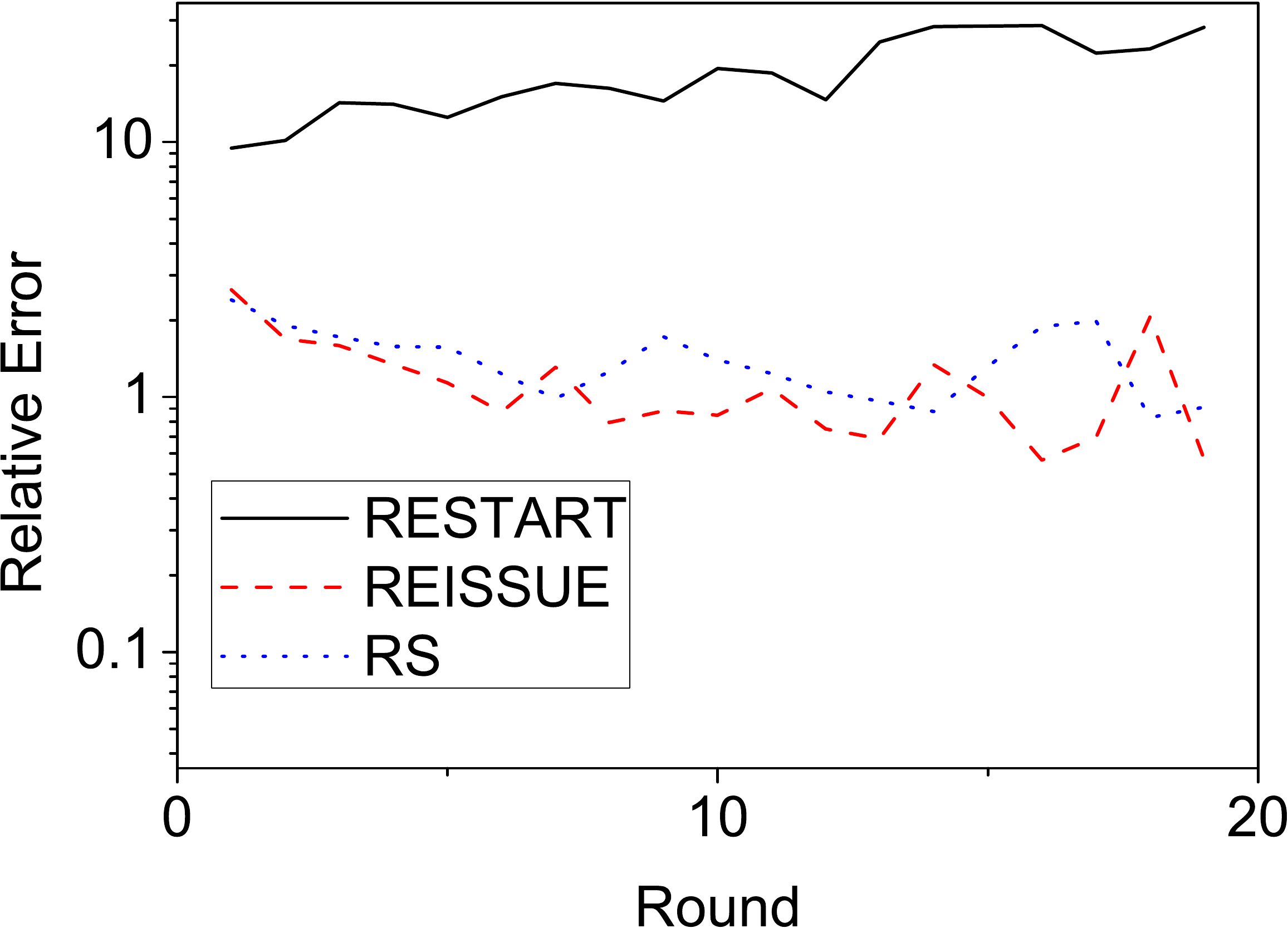}
\vspace{-2mm}\caption{Small Change}
\label{fig14}
\end{minipage}
\hspace{3mm}
\begin{minipage}[t]{0.23\linewidth}
\centering
\includegraphics[width = 40mm]{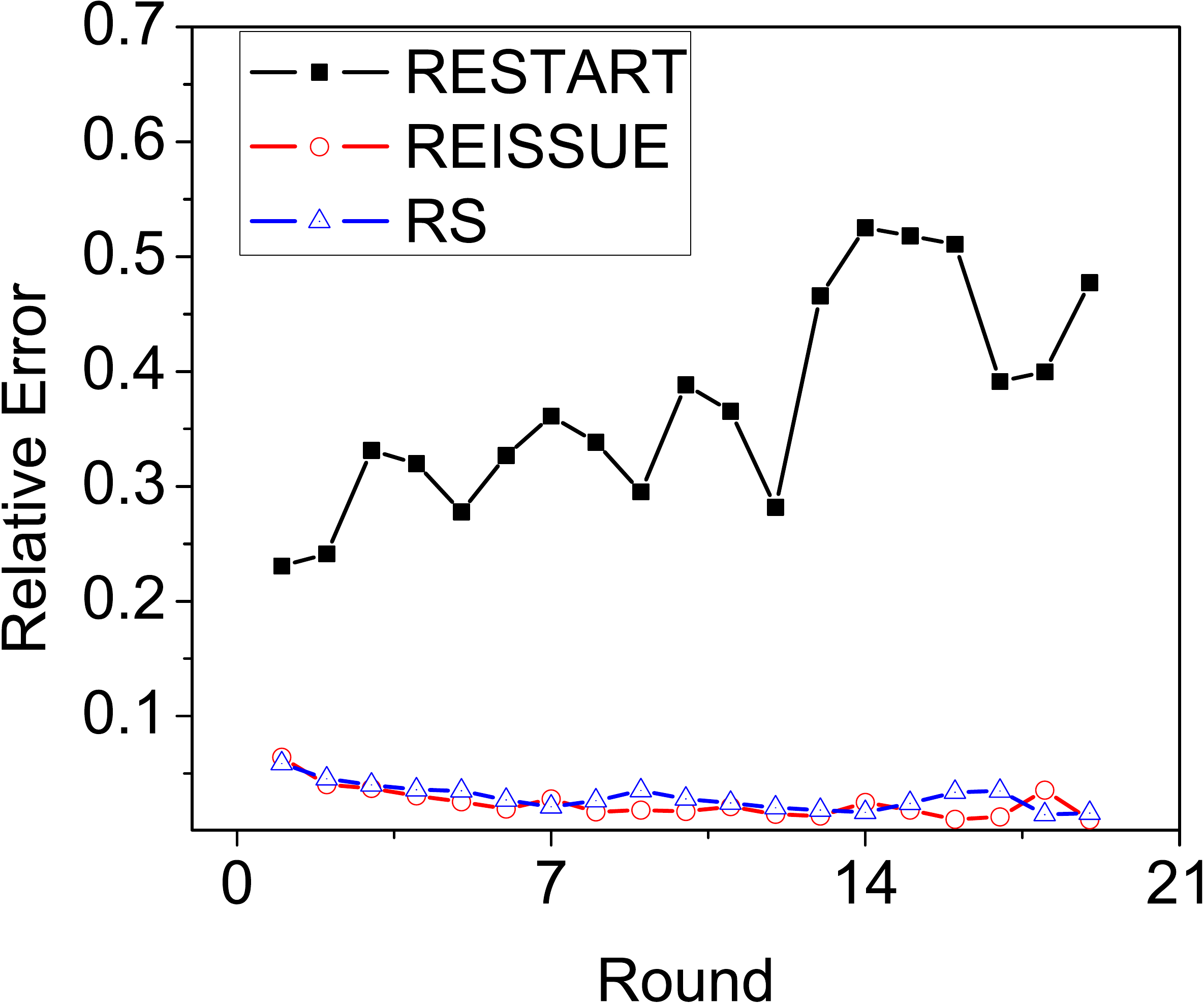}
\vspace{-2mm}\caption{Small Change (Absolute Estimation)}
\label{fig14New}
\end{minipage}
\hspace{3mm}
\begin{minipage}[t]{0.22\linewidth}
\centering
\includegraphics[width = 40mm]{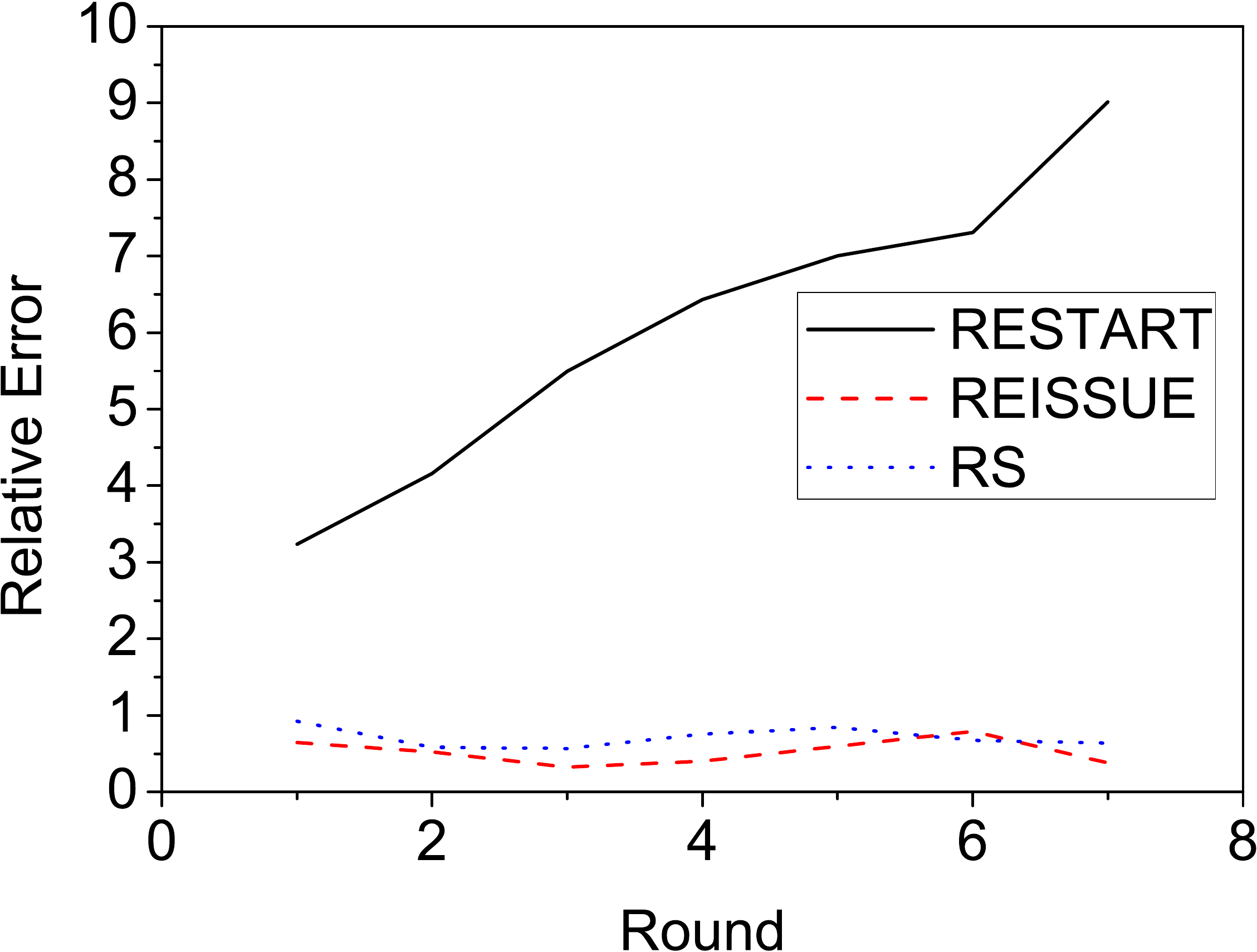}
\vspace{-5.5mm}\caption{Big Change}
\label{fig15}
\end{minipage}
\hspace{-2mm}
\end{figure*}

\begin{figure*}[ht]
\begin{minipage}[t]{0.23\linewidth}
\centering
\includegraphics[width = 38mm]{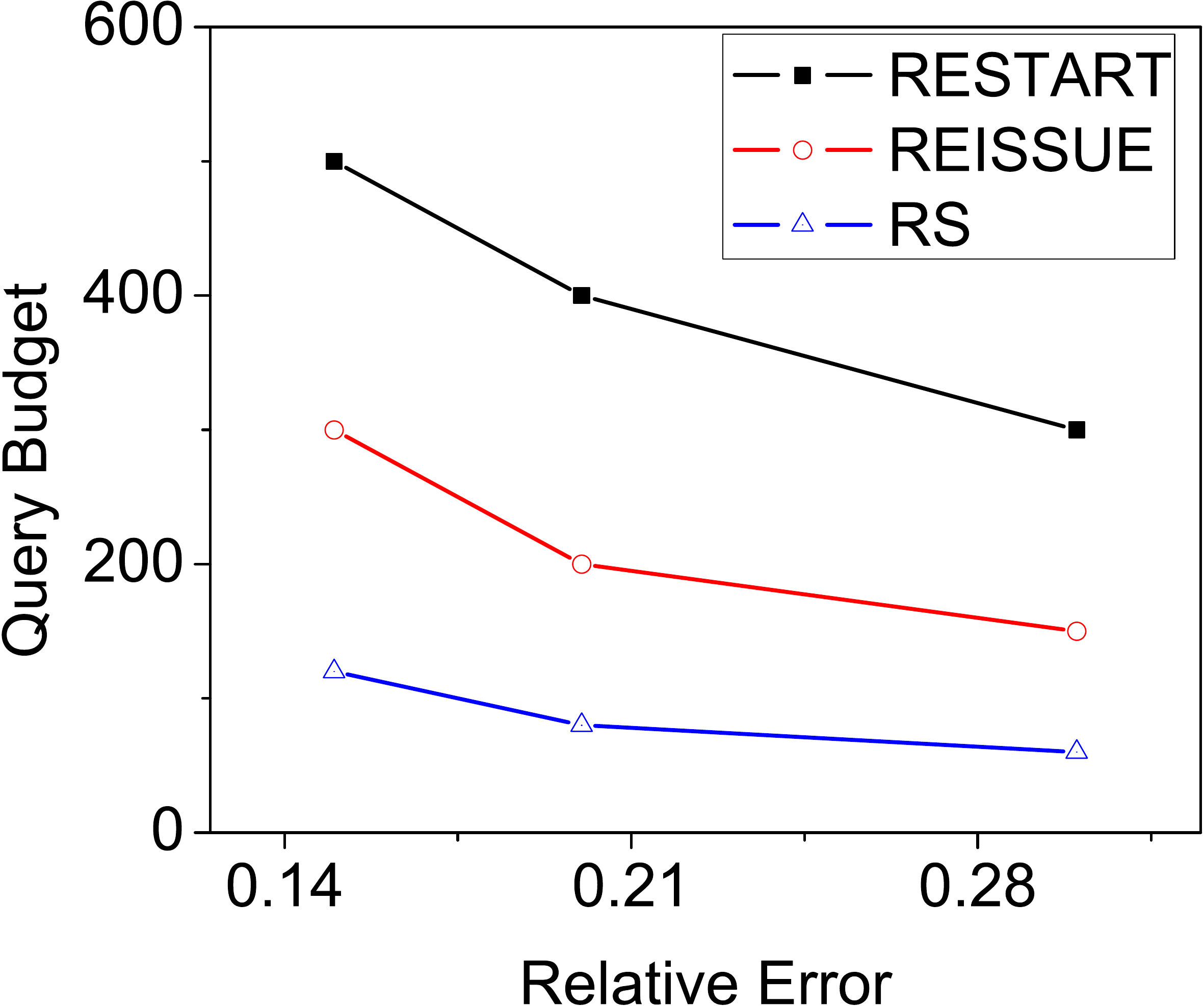}
\vspace{-2mm}\caption{Accuracy Vs Budget}
\label{fig:accuracyVsQueryBudget}
\end{minipage}
\hspace{3mm}
\begin{minipage}[t]{0.22\linewidth}
\centering
\includegraphics[width = 40mm]{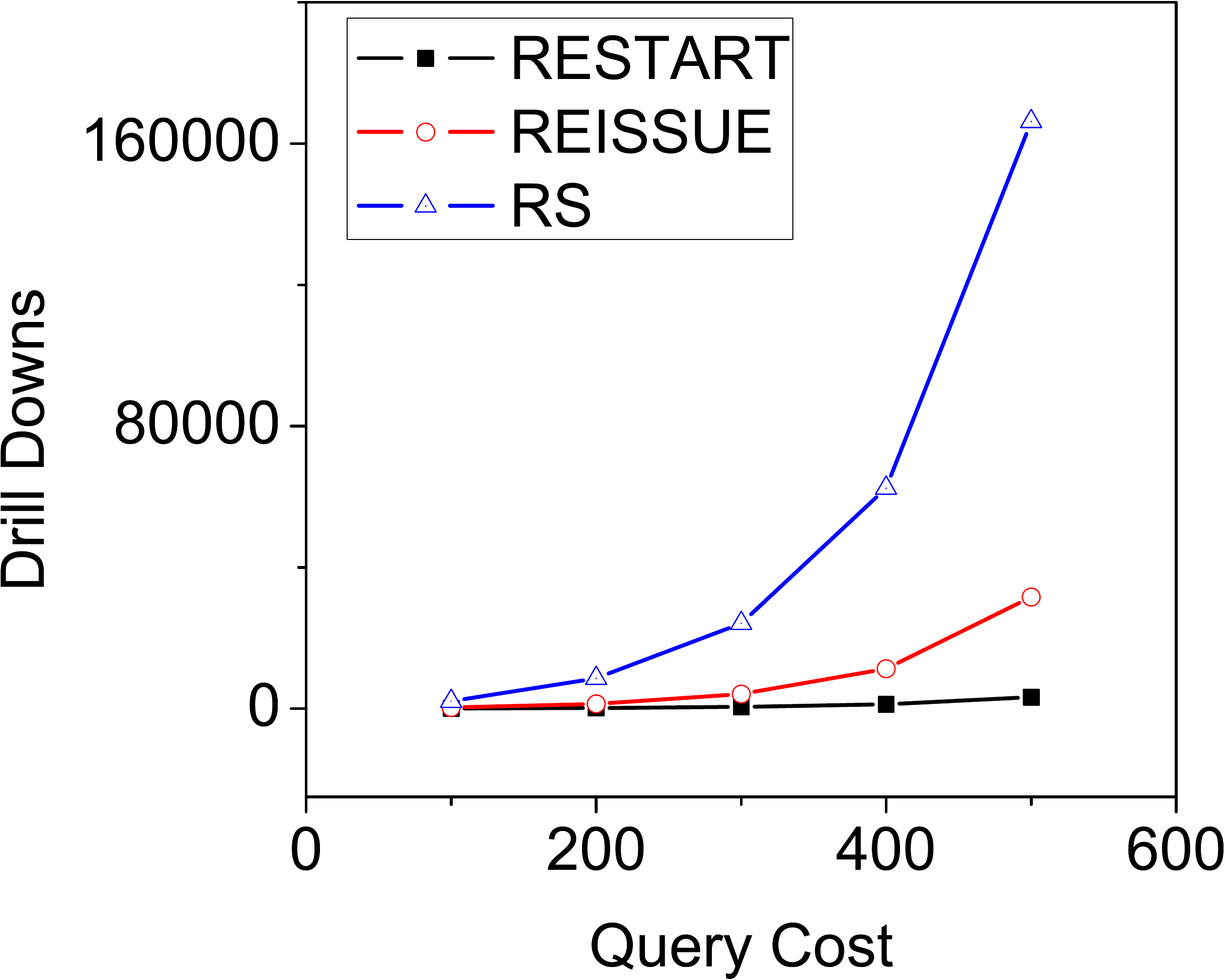}
\vspace{-5.5mm}\caption{Drill Downs}
\label{fig:drillDownsVsQC}
\end{minipage}
\hspace{3mm}
\begin{minipage}[t]{0.22\linewidth}
\centering
\includegraphics[width = 40mm]{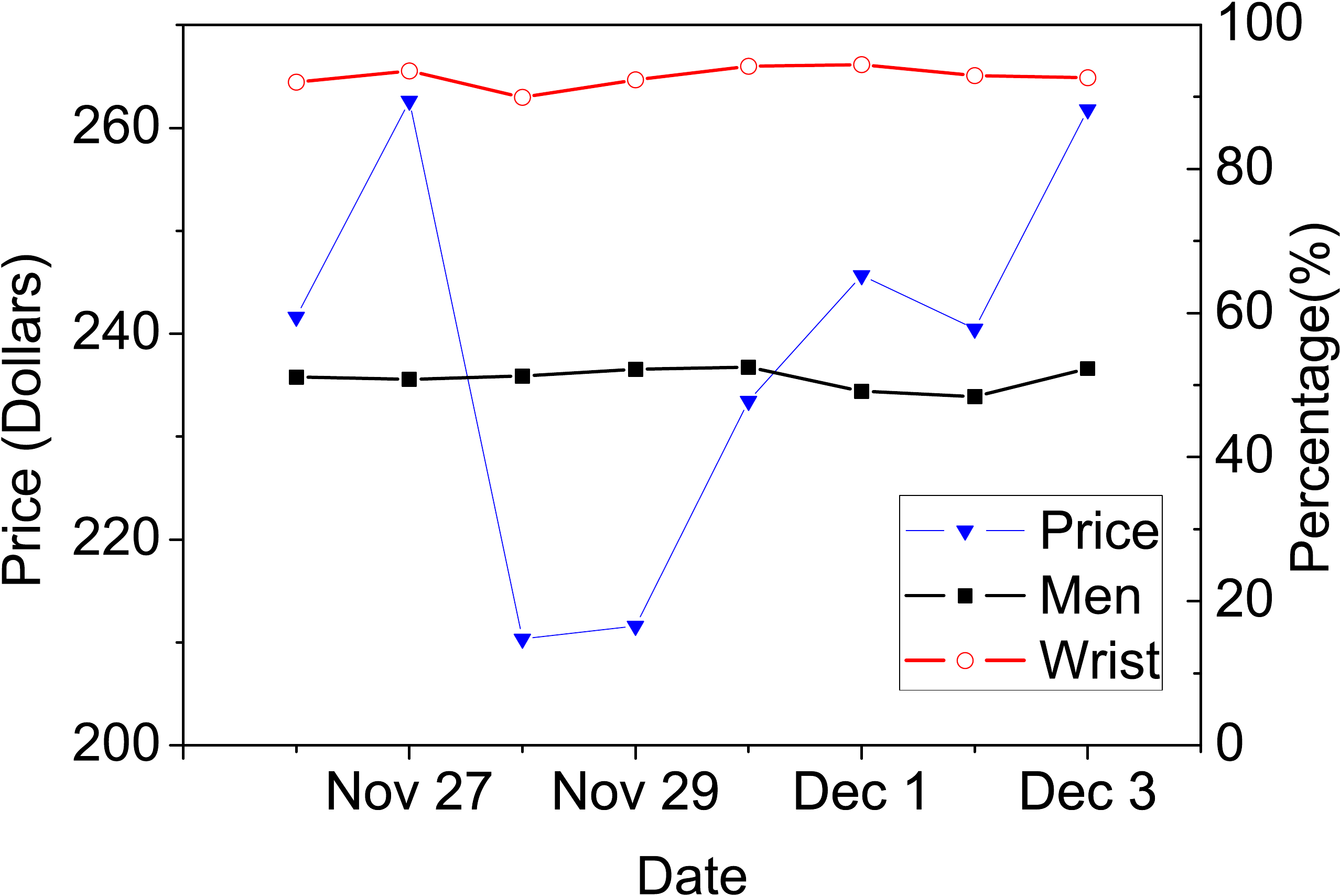}
\vspace{-5.6mm}\caption{\emph{Amazon.com}}
\label{fig16}
\end{minipage}
\hspace{3mm}
\begin{minipage}[t]{0.23\linewidth}
\centering
\includegraphics[width = 40mm]{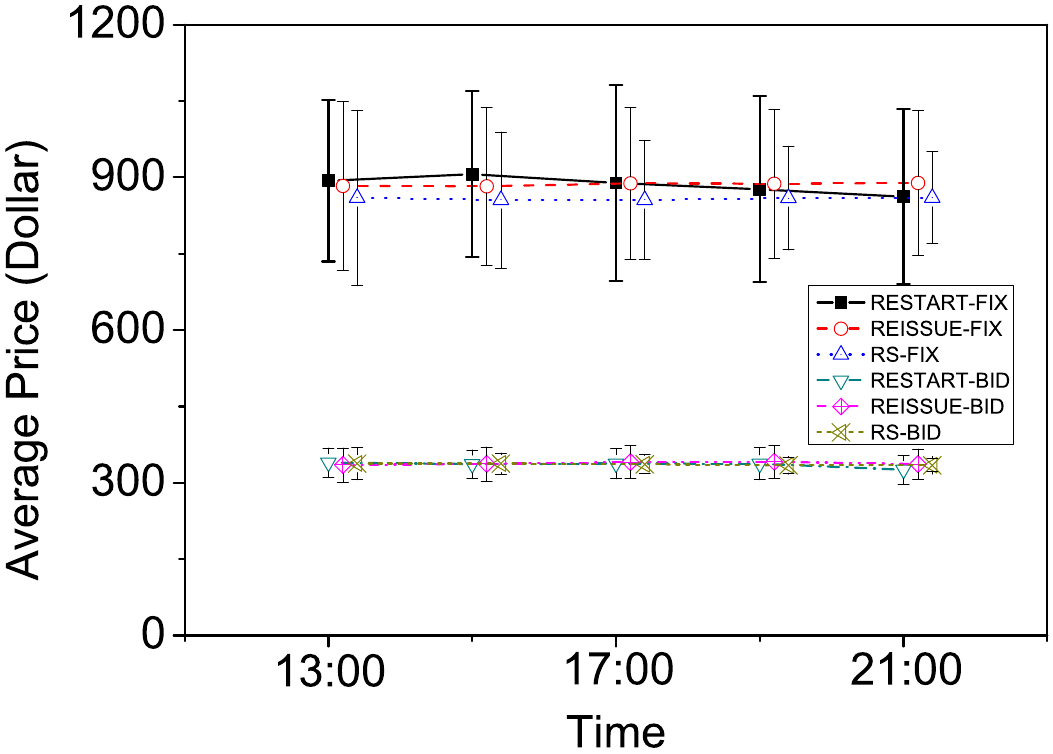}
\vspace{-2mm}\caption{eBay.com}
\label{figebay}
\end{minipage}
\hspace{-2mm}
\end{figure*}

\subsection{Experimental Results}
\noindent {\bf Single-round aggregates: } We started by comparing the performance of all three algorithms for estimating single-round aggregates over the Yahoo!~Autos dataset under the default insertion/deletion schedule and a per-round query budget of 500.  Figures \ref{fig1} depicts how the relative error on estimating COUNT(*) changes round after round, while Figure \ref{fig2} depicts the error bars of raw estimations. One can see from the figure that, as we proved in the paper, all three algorithms produce unbiased estimations, with RS having the smallest variance (i.e., shortest error bar). In contrast to RESTART, both REISSUE and RS provide more accurate estimations in latter rounds by leveraging historic query answers.

In Figure \ref{fig3}, we tested the real-world performance of REISSUE and RS when there are intra-round updates to the Yahoo!~Autos database. Specifically, we consider a worst-case scenario where our algorithms are executed every hour and take the entire hour to finish.  Within each hour, a tuple is inserted to the database every 12 seconds, while an existing tuple is deleted every 21 seconds (i.e., evenly distributed over time according to the default insertion/deletion schedule). As one can see from the figure, even in this case, the estimation accuracy of REISSUE and RS is still very close to the round-based model without any intra-round updates, consistent with our discussions in \S\ref{sec:dum}.

We then studied how the algorithms perform under extreme-case scenarios, in order to test our theoretical analysis in \S\ref{sec:rie} and \ref{sec:sam}. We started with the case when the database barely changes over time - adding one tuple per round to the real dataset. Figure \ref{fig4} depicts the results. One can see that while REISSUE significantly outperforms RESTART, the reduction of its relative error ``tapers off'' over time and remains at around 0.3. On the other hand, the estimation error produced by RS keeps decreasing (to around 0.2 after 50 rounds), consistent with our theoretical analysis in \S\ref{sec:sam}.

We also tested the worst-case scenario for our algorithms - when the database undergoes dramatic changes. Specifically, we started with 100,000 tuples in the real dataset and, for each round, added 10000 tuples and deleted 5\% of existing tuples. Figure \ref{fig5} shows that our algorithms still outperform the baseline significantly. Nonetheless, if we further change $k$ to 1, Figure \ref{fig6} shows that RESTART performs better - as predicted by
Theorem~\ref{thm:nth}.

Finally, we tested how the performance comparison changes with various parameters of the database and the search interface. Figures \ref{fig7} to \ref{fig11} tested $k$, the query budget $G$ per round, the number of inserted/deleted tuples per round, the number of attributes $m$, and the starting database size $|D_1|$, respectively. Figures \ref{fig7} and \ref{fig8} depicts the estimation error after 50 rounds on the real dataset. And in Figure \ref{fig11}, to enable the scalability test up to 10 million tuples, we set $m = 50$. In the following, we discuss a few observations from these figures:

Note from Figure \ref{fig7} that the relative error is smaller when $k$ is bigger, consistent with our theoretical analysis in \S\ref{sec:ta}. For query budget $G$, observe from Figure~\ref{fig8} that while all three algorithms produce smaller errors with a larger budget, RS remains the top choice throughout (though its advantage over REISSUE diminishes with a larger $G$ because updating only takes a small part of the budget anyway). Figure~\ref{fig9} compares cases when changes to the database range from deleting 30 (randomly chosen existing) tuples per round to inserting 30 new tuples. One can see from the figure that RS significantly outperforms RESTART in all cases. REISSUE, on the other hand, suffers when more than half of the database (30 $\times$ 100 = 3000 out of 5000) are deleted at the end of 100 rounds, consistent with the conclusion from
Theorem~\ref{thm:nth}.

Figure~\ref{fig10} shows that the performance for all algorithms is independent of $m$ - again, consistent with the theoretical results. The results from Figure~\ref{fig11} indicate a bigger relative error of RESTART for a larger database but a static accuracy of our algorithms.
Notice that the performance improvement by our algorithms {\em widens} when the database size increases.

%
Figure~\ref{fig12} depicts the relative error after 100 rounds for estimating SUM aggregates with 0, 1, 2, and 3 conjunctive selection conditions. One can see from the figures that both RS and REISSUE significantly outperform RESTART in all cases, with RS producing even smaller errors than REISSUE. Also, observe from Figure~\ref{fig12} that the more selective the aggregate is, the lower the relative error will be.

\vspace{1mm}
\noindent {\bf Trans-round Aggregates:} For trans-round aggregates, we started with testing the {\em running average count} - i.e., AVG$(|D_i|, |D_{i-1}|, \ldots)$. Figure\ref{fig13} depicts the comparison between the three algorithms over the real dataset when the COUNT of the last 2, 3, and 4 rounds are taken as inputs to the AVG function. One can see that RS has the best performance in all cases, while REISSUE and RS both significantly outperform RESTART.

We then tested in Figure \ref{fig14} another trans-round aggregate over the real-world dataset: the change of database size from last round (i.e., $|D_i| - |D_{i-1}|$). In this test, the database undergoes minor changes at each round - specifically, with 3000 tuples being inserted and 0.5\% being removed. We would like to note that this is an extreme-case scenario where the database undergoes very minor (less than 1\%)
changes, and the aggregate monitoring task tries to measure exactly how much change has occurred. Figure~\ref{fig14New} shows the absolute estimations produced by all algorithms which shows the stark superiority of our algorithms over RESTART-ESTIMATOR. In comparison, Figure \ref{fig15} depicts the same aggregate when the database is substantially changed round after round - with 10000 tuples inserted and 5\% removed. One can make a number of interesting observations from the figures: First, RS and REISSUE outperform RESTART by orders of magnitude when the database change is minor - confirming previous discussions that RESTART yields extremely poor results when the size difference is small. When the database undergoes major changes, RS and REISSUE converge to the same performance, again confirming previous analysis, and both still hold significant superiority over RESTART. Finally, the relative error produced by RESTART keeps increasing over time (with a larger database size), while REISSUE and RS have decreasing relative errors - again consistent with our theoretical analysis in \S\ref{sec:rie} and \ref{sec:sam}. 
The anomalous result (where REISSUE slightly outperforms RS) is due to the fact that the new drill downs do not reduce relative error.
According to our theoretical analysis, when the new drill downs yield much higher variance (and therefore error) than updating the old ones, the performance of RS-ESTIMATOR is essentially reduced to that of REISSUE-ESTIMATOR.

Figure~\ref{fig:accuracyVsQueryBudget} shows the efficiency benefits achieved by our algorithms. 
Instead of fixing the query budget, we ran our algorithms
till the estimation was within a relative error of 0.15, 0.2 and 0.3.
The experimental results show that lower relative errors require a higher query cost.
However, our algorithms requires substantially less queries than RESTART-ESTIMATOR to achieve same relative error.
Figure~\ref{fig:drillDownsVsQC} measures how REISSUE- and RS-ESTIMATOR achieve query savings
enabling us to perform additional drill downs for the same query budget.
This experiment was conducted over Yahoo! Autos dataset over 50 rounds.
As expected, our algorithms achieve significant savings over RESTART-ESTIMATOR due to their ability to leverage
historical query results.


\vspace{1mm}
\noindent {\bf Live Experiments:} We conducted live experiments over Amazon.com and eBay.com, respectively, in order to demonstrate the effectiveness of our algorithms in real-world settings. For Amazon.com, we conducted the experiment during Thanksgiving week, 2013, using the Amazon Product Advertising API with $k = 100$ and a query budget of 1,000 queries per day. Specifically, we monitored various aggregates over all watches sold by Amazon.com.  Admittedly, we have no access to the ground truth in this experiment, and thus cannot measure the exact accuracy of the aggregate estimations we generate. Nonetheless, as one can see from Figure~\ref{fig16}, the average price estimation we generated shows a sharp drop ($\sim$\$50) on Thanksgiving day (Nov 28) and the Black Friday (Nov 29), consistent with the common sense that most sellers are running promotions during this period.  On the other hand, two other aggregates we tracked, the percentage of watches for men and the percentage of wrist watches, barely changed during this time period, once again consistent with the common sense.

For eBay.com, we conducted the experiment on Monday, February 24, from 1pm to 9pm EST, using the eBay Finding API with $k = 100$ and a query budget of 250 queries per hour (for each algorithm we tested). Figure~\ref{figebay} depicts the experimental results. Specifically, we monitored the average (current) price of all women's wrist watches which (1) offer a ``Buy It Now'' option (i.e., with FixedPrice returned by the API, represented as -FIX in the figure), and (2)  offer a bidding option (represented as -BID in the figure). We tested in this experiment all three algorithms discussed in the paper - i.e., RESTART-, REISSUE-, and RS-ESTIMATOR.

One can make two observations from the figure. First, the average price of Buy-It-Now items is significantly higher than that of the bidding ones, consistent with the intuition that (1) Buy-It-Now items are generally more expensive, (2) while a price snapshot of an item for bid is likely lower than the final transaction price, the price snapshot of a Buy-It-Now item is usually the exact transaction price. The second observation is on the performance comparison between REISSUE-/RS-ESTIMATOR and the RESTART baseline.  One can see that, while all three algorithms perform similarly at the beginning (around 1pm), our REISSUE and RS algorithms perform significantly better than RESTART as time goes by, especially for Buy-It-Now items, consistent with our theoretical predictions. Finally, note that the performance improvement of REISSUE and RS over RESTART is more significant for Buy-It-Now items than bidding ones. This is because bidding items are updated much more frequently than Buy-It-Now items, once again consistent with our theoretical results that the less the database changes, the better REISSUE and RS will perform in comparison with  RESTART.

\vspace{-2mm}
\section{Related Work}
\noindent {\bf Information Integration and Extraction for Hidden databases:} A significant body of research has been done in this field - see tutorials \cite{CC:06, DRV:06}. Due to limitations of space, we only list a few closely-related work: \cite{RG01} proposes a crawling solution. Parsing and understanding web query interfaces was extensively studied (e.g., \cite{DKYL:09, ZHC:04}). The mapping of attributes across web interfaces was studied in \cite{HCH:04}.

\noindent{\bf Aggregate Estimations over Hidden Web Databases:}
There has been a number of prior work in performing aggregate estimation over static hidden databases.
\cite{DJJ+10} provided an unbiased estimator for COUNT and SUM aggregates for {\em static} databases with form based interfaces.
As discussed in \S\ref{sec:rie}, single-round estimation algorithms such as \cite{DJJ+10} could be adapted for dynamic databases by
treating each round as a separate static database and rerun \cite{DJJ+10} repeatedly.
However, this is possibly a wasteful approach as no information from prior invocations are reused as shown by our experimental results.
\cite{DDM:07,DZD:09,DZD:10} describe efficient techniques to obtain random samples from hidden web databases
that can then be utilized to perform aggregate estimation.
\cite{DBLP:journals/dke/AfratiLL08} proposed an adaptive sampling algorithm for answering aggregation queries
over websites with hierarchical structure.
Recent works such as \cite{LWA12, DBLP:conf/edbt/WangA11} propose more sophisticated sampling techniques so as to reduce the variance of the aggregate estimation.
For hidden databases with keyword interfaces, prior work have studied estimating the size of
search engines \cite{BB:98,BG:06,Zhang:2011:MSE:1989323.1989406}, corpus \cite{BFJ+:06} or document collection \cite{SZS+:06}.
Unlike this paper, all these prior work assume a static database.
While some of the techniques described above are applicable for estimating single round estimates per round,
our paper initiates the formal study of estimating single and trans round aggregates over a dynamic database.

\noindent{\bf Aggregate Query Processing over Dynamic Databases:}
There has been extensive work on approximate aggregate query processing over databases
using sampling based techniques 
\cite{confsigmodChaudhuriDN01,journals/tods/ChaudhuriDN07,cddmn} 
and non sampling based techniques such as histograms 
\cite{conf/icde/PoosalaG99} 
and wavelets 
\cite{chakrabarti00approximate}.
See \cite{GG:01} for a survey.
A common approach is to build a synopsis of the database or data stream and use it for aggregate estimation.
Maintenance of statistical aggregates in the presence of database updates have been considered in
\cite{Gibbons:1999:SDS:314500.315083,Gehrke:2001:CCA:375663.375665,Dobra:2002:PCA:564691.564699}.
Another related area is answering continuous aggregate queries which are evaluated continuously over stream data
\cite{BBDMW02,Babu:2001:CQO:603867.603884}.
A major difference with prior work is that the changes to underlying database is not known to our algorithm
and we could also perform trans-round aggregate estimates.

\vspace{-1mm}
\section{Conclusion And Future Work}
In this paper we have initiated a study of estimating aggregates over dynamic hidden web databases which change over time through its restrictive web search interface. We developed two main ideas: query reissuing and bootstrapping-based query-plan adjustment. We provided theoretical analysis of estimation error and query cost for the proposed ideas, and also described a comprehensive set of experiments that demonstrate the superiority of our approach over the baseline ones on synthetic and real-world datasets.
There are a number of possible future directions including:
(1) a study of how meta data such as COUNT can be used to guide the design of drill
downs in future rounds, and (2) given a workload of aggregate queries, how to minimize the
total query cost for estimating all of them, and (3) how to leverage both keyword search
and form-like search interfaces provided by many web databases to further improve the
performance of aggregate estimations.
\vspace{-1mm}

\bibliographystyle{abbrv}
\bibliography{dynamo}  

\end{document}